\documentclass{article}[14pt]

\normalsize

\usepackage{cite}
\usepackage{amsmath,amssymb,amsfonts}
\usepackage{graphicx}
\usepackage{textcomp}
\usepackage{xcolor}

\usepackage[T2A,T1]{fontenc} 
\usepackage[utf8]{inputenc}
\PassOptionsToPackage{english}{babel}
\PassOptionsToPackage{svgnames}{xcolor}
\PassOptionsToPackage{nameinlink,capitalize,noabbrev}{cleveref}
\PassOptionsToPackage{shortlabels}{enumitem}
\usepackage{%
  fontenc,inputenc,babel,%
  xstring,enumitem,
  amsmath,amssymb,amsthm,thm-restate,mathabx,mathtools,bm,%
  multirow,etoolbox,fancyhdr,lastpage,newfloat,tabularx,%
  booktabs,%
  rotating,%
  tcolorbox,caption,%
  hyperref,cleveref,%
  xspace,
}

\usepackage{algorithm}
\usepackage[noend]{algpseudocode}

\hypersetup{
  colorlinks,
  linkcolor=black,
  citecolor=black,
  urlcolor=black,
  bookmarksnumbered,
}
\newcommand{\doi}[1]{doi: \href{http://dx.doi.org/#1}{\path{#1}}}

\newtheorem{theorem}{Theorem}[section]
\newtheorem{proposition}[theorem]{Proposition}
\newtheorem{lemma}[theorem]{Lemma}
\newtheorem{corollary}[theorem]{Corollary}
\newtheorem{definition}[theorem]{Definition}
\newtheorem{remark}[theorem]{Remark}
\newtheorem{example}[theorem]{Example}

\numberwithin{equation}{section}

\newcommand\ifnull[3]{%
  \ifx\null#1%
    #2%
  \else%
    #3%
  \fi}

\newcommand{\ZZ}{\ensuremath{\mathbb{Z}}}

\newcommand{\RR}{\ensuremath{\mathbb{R}}}

\newcommand{\field}{\mathbb{F}_q}
\newcommand{\fieldc}{\widebar{\mathbb{F}}_q}
\newcommand{\ffield}{F}

\newcommand{\places}{\mathbb{P}}

\renewcommand{\vec}[1]{\bm{#1}}
\newcommand{\mat}[1]{\vec{#1}}
\renewcommand{\L}{\mathcal{L}}
\newcommand{\dimL}{l}
\newcommand{\Pinf}{P_{\infty}}
\newcommand{\val}[1][\null]{\ifnull{#1}{v_{\Pinf}}{v_{#1}}\xspace}

\newcommand{\code}{\mathcal{C}_{\L}(D,G)}
\renewcommand{\r}{\vec{r}}

\newcommand{\errors}{\mathcal E}
\newcommand{\M}{\mathcal M_{s,\ell}(D,G)}
\newcommand{\y}[1][\null]{\ifnull{#1}{y}{y^{(#1)}}}
\newcommand{\g}[1]{g^{(#1)}}

\renewcommand{\H}{\mathcal{H}}

\renewcommand{\S}{\mathcal{S}}

\newcommand{\pow}[1]{\widehat{#1}}

\newcommand{\A}{\mathcal{A}}
\newcommand{\Y}{\mathcal{Y}}

\renewcommand{\and}{\quad \text{and} \quad}
\newcommand{\for}{\quad \text{for }}

\newcommand{\assign}{\leftarrow} 
\newcommand{\Input}{\Statex \textbf{Input: }} 
\newcommand{\Output}{\Statex \textbf{Output: }} 

\newcommand{\algo}[1]{\ensuremath{\mathsf{#1}}\xspace}
\newcommand{\Decode}{\algo{Decode}}
\newcommand{\Evaluate}{\algo{Evaluate}}
\newcommand{\Interpolate}{\algo{Interpolate}}
\newcommand{\GenYa}{\algo{Generators}_{\Ya}}
\newcommand{\BasisProducts}{\algo{BasisProducts}_{\field[x]}}
\newcommand{\GenFqx}{\algo{Generators}_{\field[x]}}
\newcommand{\RootFind}{\algo{RootFinding}}

\makeatletter
\def\easycyrsymbol#1{\mathord{\mathchoice
  {\mbox{\fontsize\tf@size\z@\usefont{T2A}{cmr}{m}{n}#1}}
  {\mbox{\fontsize\tf@size\z@\usefont{T2A}{cmr}{m}{n}#1}}
  {\mbox{\fontsize\sf@size\z@\usefont{T2A}{cmr}{m}{n}#1}}
  {\mbox{\fontsize\ssf@size\z@\usefont{T2A}{cmr}{m}{n}#1}}
}}
\makeatother
\newcommand{\Ya}{\easycyrsymbol{\CYRYA}}

\newcommand{\bigO}{\mathcal{O}}
\newcommand{\softO}{\tilde \bigO}
\newcommand{\N}{\mathcal{N}}
\newcommand{\V}{\mathcal{V}}

\DeclareMathOperator{\supp}{supp}

\DeclareMathOperator{\ev}{ev}

\DeclareMathOperator{\rem}{ \ rem \, }
\DeclareMathOperator{\Con}{Con}

\begin{document}

\title{Fast Decoding of AG Codes}

\author{Peter~Beelen,
        Johan~Rosenkilde,
        and~Grigory~Solomatov.
      }

\maketitle

\begin{abstract}
  We present an efficient list decoding algorithm in the style of Guruswami-Sudan for algebraic geometry codes. Our decoder can decode any such code using $\softO(s\ell^{\omega}\mu^{\omega-1}(n+g))$ operations in the underlying finite field, where $n$ is the code length, $g$ is the genus of the function field used to construct the code, $\ell$ is the designed list size and $\mu$ is the smallest positive element in the Weierstrass semigroup at some chosen place; the ``soft-O'' notation $\softO(\cdot)$ is similar to the ``big-O'' notation $\bigO(\cdot)$, but ignores logarithmic factors. For the interpolation step, which constitutes the computational bottleneck of our approach, we use known algorithms for univariate polynomial matrices, while the root-finding step is solved using existing algorithms for root-finding over univariate power series.
\end{abstract}


\section{Introduction}

Containing some of the best error-correcting codes currently known, algebraic geometry (AG) codes have received a lot of attention since their introduction by Goppa in \cite{goppa_algebraico-geometric_1983}. The celebrated Guruswami-Sudan decoder \cite{guruswami1998improved} for these codes relies on an interpolation step as well as a root-finding step and is capable of decoding beyond half the designed minimum distance by returning a list of all codewords within a certain Hamming distance $\tau$ from the received word. In this article, we present an efficient realization of this decoder, achieving the best known complexity in the fully general setting of arbitrary AG codes. Moreover, except for the particularly simple case of Reed-Solomon codes, our decoder is at least as fast as all existing decoders which are tailored for specific families of codes. This article is based on a chapter of the PhD thesis of the third author \cite{solomatov2021thesis}.

Following the common practice, we will measure algorithmic complexity by asymptotically upper-bounding the number of arithmetic operations in the underlying finite field $\field$, relying on the \emph{big-O} notation $\bigO(\cdot)$ as well as the \emph{soft-O} notation $\softO(\cdot)$, which ignores logarithmic factors. Formally, $\softO(h) = \bigcup_{j=0}^\infty \bigO(h \log(h)^j)$ for any function $h : \RR_{\ge 0} \to \RR_{\ge 0}$, where $\RR_{\ge 0}$ denotes the set of non-negative real numbers. Analogously to $\RR_{\ge 0}$, we will also write $\ZZ_{\ge 0}$ and $\ZZ_{> 0}$ for the non-negative and the positive integers respectively. Our complexity estimates will also involve $\omega$, which denotes some real number such that the product of any two matrices in $\field^{m \times m}$ can be computed using $\bigO(m^\omega)$ operations in $\field$. The naive algorithm for matrix multiplication yields $\omega = 3$, and it is clear that $\omega \ge 2$ in general; the current record with $\omega < 2.37286$ is due to \cite{alman2021refined}.

Our decoder has the complexity $\softO(\ell^{\omega+1}\mu^{\omega-1}(n+g))$, however, in a series of remarks throughout the article we explain how it can be slightly improved to $\softO(s\ell^{\omega}\mu^{\omega-1}(n+g))$; here $n$ is the code length, $g$ is the genus of the function field used to construct the code, $\ell$ is the designed list size, $s \le \ell$ is the multiplicity, and $\mu$ is the smallest element in the Weierstrass semigroup of some rational place $\Pinf$ which is not one of the evaluation places. As we will see in \cref{sec:simpler-setting}, the existence of such $\Pinf$ can be assumed without any loss of generality.

\subsection{Related work}

As mentioned earlier, the paradigm of Guruswami-Sudan list decoding revolves around two main steps: interpolation and root-finding. As former is generally more computationally demanding, it has historically received the most attention.
Several authors, including \cite{nielsen2000decoding, o2002grobner, mceliece_guruswami-sudan_2003, alekhnovich_linear_2005, farr2005grobner}, formulated the interpolation step as a problem of finding a polynomial, minimal with respect to a weighted monomial order, in a certain vanishing ideal. Prompted by this, Lee and O'Sullivan developed a technique for obtaining such a polynomial from a Gr\"obner basis (of  $\field[x]$-modules), that was itself computed starting from a particular generating set\,--\,first for RS codes \cite{lee_list_2008}, and then for one-point Hermitian codes \cite{lee_list_2009}. The complexity of this strategy was further improved by Beelen and Brander in \cite{beelen_efficient_2010} by utilizing Alekhnovich's algorithm for row reduction of polynomial matrices \cite{alekhnovich_linear_2005}. Furthermore, their decoder was applicable to the wider family of one-point codes over $C_{ab}$ curves, making it more general. Specializing back to one-point Hermitian codes, Rosenkilde and Beelen \cite{nielsen_sub-quadratic_2015} sped up this approach even more by delegating the row-reduction phase to the algorithm by Giorgi, Jeannerod and Villard \cite{giorgi_complexity_2003}, which is more efficient than the one by Alekhnovich. Doing this required additional improvements to keep up with the new target complexity, including efficient computation of the initial $\field[x]$-basis, as well as a way of handling fractional weights. The result was the first list-decoder of one-point Hermitian codes having sub-quadratic complexity in the code length.
In the current article, we generalize the tools from \cite{nielsen_sub-quadratic_2015} to be applicable to all AG codes, relying on the conceptual framework from \cite{lee_unique_2014} to represent function field elements using \emph{Ap\'ery systems}. 

Before shifting our attention to the root-finding step, we ought to mention the multivariate interpolation algorithm by Chowdhury, Jeannerod, Neiger, Schost and Villard \cite{chowdhury_faster_2015}\,--\,it was the first to enable the currently best complexity in the special case of RS codes, albeit in a probabilistic manner. A deterministic algorithm with the same complexity was later given in \cite{jeannerod_computing_2017}.

Some of the earliest root-finding algorithms for Guruswami-Sudan list-decoding include Roth and Ruckenstein's \cite{roth_efficient_2000} as well as Gao and Shokrollahi's \cite{gao2000computing}. Alekhnovich described in \cite{alekhnovich_linear_2005} an efficient approach for computing the $\field[\![x]\!]$-roots modulo $x^\beta$ of a polynomial $Q \in \field[\![x]\!][z]$; its complexity was shown in \cite{nielsen_sub-quadratic_2015} to be $\softO(\beta^2 \ell)$ operations in $\field$, where $\ell$ is the $z$-degree of $Q$. Another technique by Berthomieu, Lecerf and Quintin \cite{berthomieu_polynomial_2013} achieved the cost $\softO(\beta \ell^2)$. In this article, we rely on the algorithm by Neiger, Rosenkilde and Schost \cite{neiger_fast_2017}, whose complexity of $\softO(\beta \ell)$ operations is provably quasi-optimal.

The complexity of our decoder is at least as good as, and often faster than, the complexity of previous decoders based on the Guruswami-Sudan paradigm. As far as we know, there is only one exception: in the case of RS codes, the complexity of the algorithms from \cite{chowdhury_faster_2015} or \cite{jeannerod_computing_2017} is a factor of $\ell/s$ better. To illustrate the strength and versatility of our results, in \cref{sec:decoding}, examples are given of the list decoding of various families of AG codes. One further remark should be made, namely the case of bounded distance decoding. Setting $s = \ell = 1$ and assuming that $g \in \bigO(n)$, the complexity our decoder simplifies to $\softO(\mu^{\omega-1}n)$. In this case, the decoder can always correct up to $(d^*-1-g)/2$ errors, where $d^*$ denotes the designed minimum distance of an AG code, also known as the Goppa bound. The same decoding radius is achieved in \cite{sakata_fast_2014-1} with complexity $\bigO(\mu n^2)$. Since $\mu \le g$ and we assumed $g \in \bigO(n)$, our complexity is better. However, Sakata's extension of the Berlekamp-Massey decoder \cite{sakata_extension_1990, sakata_generalized_1995} yields a decoding algorithm able to correct up to at least $(d^*-1)/2$ errors. In \cite{hoholdt_algebraic_1998}, the complexity in the case of certain so-called one-point AG codes is $\bigO(\mu n^2+q^{t+1}(a_1+\cdots+a_t)+tnq^t)$, where $a_1,\dots,a_t$ form a minimal set of generators of the Weierstrass semigroup at $P_\infty$. To achieve the same decoding radius with our decoder, we could choose $s$ and $\ell$ in $\bigO(g)$, but doing so might not be as efficient, since our complexity would then increase by a factor of $\bigO(g^{\omega + 1}) \subseteq \bigO(n^{\omega + 1})$.

\subsection{Strategy outline and contributions}

With the aim of making the exposition easier in the subsequent sections of the article, we now present an overview of the main steps in the proposed decoder. This consists of a way of simplifying the general setting as well as a way of efficiently carrying out the classical steps of interpolation and root-finding. The complete decoder is presented in \cref{sec:decoding}, where it is also exemplified for special cases of AG codes.

\begin{itemize}
\item{\bf Simplified setting:} In \cref{sec:simpler-setting}, we show how extending the constant field $\field$ allows us to make certain simplifying assumptions. The important takeaway here is that no generality is sacrificed in the process, while only a minor penalty is introduced into the computational complexity. In return, we may assume existence of certain rational places, as well as existence of a special function field element $x$ with controlled zeroes and poles. Having access to additional rational places is useful for a variety of reasons, among which is efficient multiplication of function field elements in a pointwise manner; the carefully chosen function $x$ acts as a fundamental building block in the way we represent function field elements.
\item{\bf Interpolation step:} This is the most involved part of the article and requires all of the computational tools from \cref{sec:algorithms}\,--\,except for \cref{subsec:findroots}, which deals with root-finding. In \cref{subsec:structure-of-M-as-Ya-module}, it is explained how the interpolation step can be viewed as a problem of finding a ``small'' element $Q$ in a certain interpolation module whose underlying ring consists of all functions that have no poles except for possibly at a fixed rational place $\Pinf$. This ring, denoted by $\Ya$, is itself a free module over $\field[x]$, which essentially means that we can represent everything as tuples of univariate polynomials. The computational path for obtaining $Q$ boils down to first constructing a generating set of the interpolation module over $\Ya$, then expanding this to a generating set over $\field[x]$, and finally, using efficient algorithms for matrices over $\field[x]$ to reduce this generating set to a ``small'' basis that contains a satisfactory $Q$.
\item{\bf Root-finding step:}
  Structurally, the obtained $Q$ is a univariate polynomial whose coefficients are function field elements; and according to the Guruswami-Sudan paradigm, list decoding reduces to finding the roots of this polynomial. We accomplish this by expressing the coefficients of $Q$ as power series in $x$, which is always possible in our simplified setting in which $x$ is a local parameter of some appropriate rational place. An existing algorithm for root-finding over the ring of power-series is then used to obtain the sought roots, albeit represented as power series; the final step of our decoder therefore consists of converting these roots back to the original representation as well as discarding any potential ``spurious'' solutions. All of this is detailed in \cref{subsec:findroots}.
\end{itemize}

Our decoder relies on a mixture of new and existing results; the novel contributions include:

\begin{itemize}
\item reduction of the fully general setting to a simpler one (\cref{sec:simpler-setting}),
\item an algorithm for encoding general AG codes with complexity $\softO(\mu n)$ (\cref{subsec:mpe}),
\item an interpolation algorithm with complexity $\softO(\mu^{\omega-1}(n+g))$ (\cref{subsec:interpolation}),
\item a root-finding algorithm with complexity $\softO(\ell^2\mu^{\omega-1}(n+g))$ (\cref{subsec:findroots}),
\item an algorithm for computing an $\field[x]$-basis of $\langle h \rangle_{\Ya}$ for any function $h$ (\cref{subsec:genFqx}).
\end{itemize}
Not counting precomputation, all of the algorithms above are sufficiently efficient to reach our target cost. Although the cost of precomputation has not been investigated in detail, it is not expected to be much more expensive than that of Gaussian elimination. A list of all precomputed objects can be found in \cref{sec:decoding}.


\section{Preliminaries}\label{sec:preliminaries}

\subsection{AG codes}
Let $\field$ be a finite field with $q$ elements, where $q$ is a power of a prime number $p$. Further, let $\ffield$ be a function field of genus $g$ and full constant field $\field$. As is common, we denote by $\places_{\ffield}$ the set of places of $\ffield$. 

For any divisor $A=\sum_i n_i Q_i$ of $\ffield$, we denote by $\supp(A)$ the support of $A$, which consists of all places $Q_i$ such that $n_i \neq 0$. A divisor $A$ is called effective, denoted by $A \ge 0$, if for all $i$ it holds $n_i \ge 0$. Further, the degree of $A$, is defined as $\deg(A)=\sum_i n_i \deg(Q_i),$ where $\deg(Q_i)$ denotes the degree of the place $Q_i$.

The well-known Riemann-Roch space of a divisor $A$ is given by
\[
  \L(A)=\{f \in \ffield\setminus\{0\} \mid (f)+A \ge 0\} \cup \{0\} \ ,
\]
where $(f)$ denotes the divisor of $f$. The Riemann-Roch space $\L(A)$ is a vector space over $\field$, whose dimension will be denoted by $\dimL(A)$. The theorem of Riemann-Roch \cite[Theorem 1.5.15]{stichtenoth_algebraic_2009} implies that $\dimL(A) \ge \deg(A)+1-g$ and that equality holds if $\deg(A) \ge 2g-1$. Moreover $\dimL(A)=0$ if $\deg(A)<0$ since the degree of a principal divisor is zero.
\begin{definition}
  \label{def:ev}
Assume that $\ffield$ has at least $n$ rational places, say $P_1,\dots,P_n$ and write $D=P_1+\cdots + P_n$. Further, let $G$ be a divisor of $\ffield$ such that $\supp(G) \cap \supp(D)=\emptyset$\footnote{The assumption that $\supp(G) \cap \supp(D)=\emptyset$ can be removed \cite[Remark 2.2.15]{stichtenoth_algebraic_2009}, but doing so does not give new AG codes up to monomial equivalence.}.
Then we define
\[
  \code = \{
  \ev_D(f) \mid f \in \L(G) \} \subset \field^n \ ,
\]
where for any $f \in \L(G)$, $\ev_D(f) = (f(P_1),\dots,f(P_n)) \in \field^n$.
\end{definition}
For future reference, we state some properties of this code, see \cite[Chapter 2]{stichtenoth_algebraic_2009} for details. First of all, it is well-known that this code has minimum distance at least $n-\deg(G).$ Since the kernel of $\ev_D$ is $\L(G-D)$, the dimension of the code equals $\dimL(G)-\dimL(G-D)$. In particular, $\code$ is the zero code if $\deg(G) <0$. Further, using the theorem of Riemann-Roch, we see that $\dim(\code)=n$, i.e. $\code=\field^n$, whenever $\deg(G) \ge n+2g-1$. Because of this, we may assume
\begin{equation}\label{eq:degree_bounds_G}
0 \le \deg(G) \le n+2g-1.
\end{equation}

\begin{remark}
  In his original construction, Goppa considered AG codes $\mathcal{C}_{\Omega}(D,G)$ defined using residues of certain differentials. These codes can also be obtained as evaluation codes \cite[Proposition 2.2.10]{stichtenoth_algebraic_2009}. Hence our decoder can also handle codes of the form $\mathcal{C}_{\Omega}(D,G)$.
\end{remark}


\subsection{Reduction to a simpler setting}
\label{sec:simpler-setting}

In this subsection, we will show that without significant increase of decoding complexity, we can assume several things about the function field $\ffield$ and the AG code $\code$ that will make the exposition of our decoding algorithm simpler later on. For example, it will be convenient to have an additional rational place $\Pinf$ of $\ffield$ that is not used in the evaluation map $\ev_D$. In fact, for some of our later algorithms, it will be convenient to have additional rational places as well. An easy way out is to increase the constant field $\field$ to $\mathbb{F}_{q^e}$ for some small value of $e$, thus introducing new rational places that can be used as additional rational places. We will denote by $F\mathbb{F}_{q^e}$, the function field obtained from $\ffield$ by extending the constant field to $\mathbb{F}_{q^e}$.

As far as decoding is concerned, the AG code $\code$ is in a trivial way a subcode (not $\mathbb{F}_{q^e}$-linear, but $\field$-linear) of the AG code obtained from the function field $F\mathbb{F}_{q^e}$ using the divisors $\Con(D)$ and $\Con(G)$, where $\Con$ denotes the conorm with respect to $F\mathbb{F}_{q^e}/F$, \cite[Definition 3.1.8]{stichtenoth_algebraic_2009}. Since all places in $\supp(D)$ are rational, we may with slight abuse of notation write $\Con(D)=D$. Hence if for a given $\tau$, one has a list decoding algorithm for $\mathcal{C}_{\mathcal L}(D,\Con(G))$ that produces all codewords at distance at most $\tau$ from a received word, one immediately obtains a list decoding algorithm for $\code$. However, since multiplication of two elements in $\mathbb{F}_{q^e}$ can be done in $\softO(e)$ operations in $\field$ \cite{cantor_fast_1991}, the value of $e$ should be small for complexity reasons. Therefore we now give a series of lemmas, each aiming to show that for small $e$, simplifying assumptions can be made about the function field $\ffield$ and the code $\code$.

\begin{lemma}\label{lem:more_rational_places}
Let $\ffield$ be a function field over $\field$ of genus $g$ and denote by $N_e$ the number of rational places of the function field $F\mathbb{F}_{q^e}$ over $\mathbb{F}_{q^e}$. If $N,e \in \ZZ_{>0}$ are such that $e \geq 2\log_q \max\{N,2g + 1\}$, then $N_e > N$. 
\end{lemma}
\begin{proof}
  The Hasse-Weil bound $|(q^e + 1) - N_e| \leq 2 q^{e/2} g$ implies that
  \[
    \log_q N_e > \log_q (q^e - 2 q^{e/2} g) = e/2 + \log_q(q^{e/2} - 2g) \geq e/2 \geq \log_qN \ .
  \]
\end{proof}

Now we show that if the function field $\ffield$ has sufficiently many rational places, then one of them has particularly simple local parameter. Recall that a function $f \in \ffield$ is called a local parameter for a place $P$ if $v_P(f)=1$, where $v_P$ denotes the valuation at $P$.

\begin{lemma}\label{lem:x_local_parameter}
Let $\ffield$ be a function field over $\field$ of genus $g$ having a rational place $\Pinf$. Let $\mu$ be the smallest positive element from the Weierstrass semigroup of $\Pinf$. Any set containing at least $3g+1$ rational places distinct from $\Pinf$, contains a place $P_0$ with local parameter from $\L(\mu\Pinf)$.
\end{lemma}
\begin{proof}
Let $x \in \L(\mu\Pinf)$ be a function satisfying $v_{\Pinf}(x)=-\mu$.
First of all, note that the extension $\ffield/\field(x)$ is separable. Indeed, assume that $\ffield/\field(x)$ is inseparable. Then by the general theory of inseparable extensions, we can find an intermediate field $E$ such that $E/\field(x)$ is separable and $\ffield/E$ is purely inseparable. Then by \cite[Proposition 3.10.2]{stichtenoth_algebraic_2009}, $\field(x) \subseteq \ffield^p=\{f^p \mid f \in F\}$, where $p$ is the characteristic. Hence $x=y^p$ for some $y \in \ffield.$ Since $x$ has a pole at $\Pinf$ only of order $\mu$, this would imply that the function $y$ also has a pole at $\Pinf$ only of order $\mu/p$. This gives a contradiction with the minimality of $\mu$.

The above implies that the Hurwitz genus formula (see for example \cite[Corollary 3.4.14]{stichtenoth_algebraic_2009}) applies to the extension $\ffield/\field(x)$. To prove the lemma, we estimate the genus of $\ffield$ using this formula. 
We write $Q_\infty=\Pinf \cap \field(x)$, which is the place at infinity of the rational function field $\field(x)$. Since $v_{\Pinf}(x)=-\mu$ and $v_{Q_\infty}(x)=-1$, we see that $e(\Pinf|Q_{\infty})=\mu$. Now suppose we have $N>3g$ rational places distinct from $\Pinf$, say $P_1,\dots,P_N$. We write $Q_i=P_i \cap \field(x)$ for their restrictions to $\field(x)$. For these rational places, we have $$v_{P_i}(x-x(P_i))=e(P_i|Q_i)v_{Q_i}(x-x(P_i))=e(P_i|Q_i).$$
Suppose that for every rational place $P_i$ it holds that $e(P_i | Q_i) \geq 2$.

Since $\mu = [\ffield : \field(x)]$ by \cite[Theorem 1.4.11]{stichtenoth_algebraic_2009}, the Hurwitz genus formula combined with the estimate $d(P_i|Q_i) \ge e(P_i|Q_i)-1$ implies that
  \begin{align*}
    2g - 2
    &\ge -2 [\ffield : \field(x)] + d(\Pinf|Q_\infty)+\sum_{i=1}^N d(P_i|Q_i) \\
    &\geq -2 \mu + (\mu - 1) + N(2-1) \ .
  \end{align*}
Since $\mu \leq g + 1$, we conclude that $N \leq 3g$, a contradiction.
Hence for one of the places $P_1,\dots,P_N$ we have $v_{P_i}(x-x(P_i))=1.$
\end{proof}

To motivate \cref{lem:x_local_parameter}, recall from \cref{lem:more_rational_places} that by extending our base field $\field$ to $\mathbb{F}_{q^e}$, we can easily ``create'' as many new rational places as we need without compromising our target complexity. By doing this, we can ensure that there exists a function $x \in \L(\mu \Pinf)$ which is also a local parameter of some rational place $P_0$ not in $\supp(G)$. As we will see in Section \ref{sec:algorithms}, membership of $x \in \L(\mu \Pinf)$ allows us to impose an $\field[x]$-module structure on the interpolation step of Guruswami-Sudan decoding. In Section \ref{sec:algorithms} we will use the assumption that $x$ is a local parameter of $P_0$ to represent certain functions in $F$ as power series in $\field[\![x]\!]$, which allows us to solve the root-finding step efficiently.

Next we consider a lemma showing that we can assume that the divisor $G$ used to define the AG code $\code$ is effective unless the code is degenerate. We call a code $\mathcal C$ degenerate if there exists $i$ such that $c_i=0$ for any codeword $c=(c_1,\dots,c_n) \in \mathcal C$. In particular the trivial code containing only the zero codeword is degenerate.

\begin{lemma}\label{lem:G_effective}
Let the function field $\ffield$ and divisors $G$ and $D$ be as before. Then either, the AG code $\code$ is degenerate or $\mathcal{C}_{\mathcal L}(D,\Con(G))$ is monomially equivalent over $\mathbb{F}_{q^e}$ with $e \ge 1+\lceil\log_q(n)\rceil$, to an AG code $\mathcal{C}_{\mathcal L}(D,G^\prime)$, where $G^\prime$ is an effective divisor of $F\mathbb{F}_{q^e}$ of degree $\deg(G)$.
\end{lemma}

\begin{proof}
Consider the finite field extension $\mathbb{F}_{q^e}/\field$ and for convenience, let us write $\mathcal C=\mathcal{C}_{\mathcal L}(D,\Con(G))$ as well as $\mathcal C_i =\{c \in \mathcal C \mid c_i=0\}.$ If $\code$ is nondegenerate, then so is $\mathcal C$. In this case $\mathcal C \neq \mathcal C_i$ for all $i$.
If every codeword in $\mathcal C$ has at least one zero coordinate, then $\mathcal C = \bigcup_{i=1}^n \mathcal C_i$, which implies that $(q^e)^k \le n (q^e)^{k-1}$ with $k=\dim \mathcal{C}.$ We see that in this case $q^e \le n$, implying $e \le \log_q(n).$ This contradiction shows that $\mathcal C$ contains a codeword of full Hamming weight $n$, say $c=\ev_D(\tilde{f})$ for some $\tilde{f} \in \L(\Con(G))$. Since by construction $\tilde{f}(P_i) \neq 0$ for all $i$, we see that the codes  $\mathcal C$ and $\mathcal{C}_{\mathcal L}(D,\Con(G)+(\tilde{f}))$ are monomially equivalent using the map $(c_1,\dots,c_n) \mapsto (\tilde{f}(P_1)c_1,\dots,\tilde{f}(P_n)c_n)$. Note that the divisor $G^\prime=\Con(G)+(\tilde{f})$ is effective and has support disjoint from $D$.
\end{proof}

Degenerate codes are not very interesting from the error-correcting point of view. Indeed, if the $i$-th coordinate of all codewords is zero, it is trivial to correct errors in that position. On the other hand that position does not carry any information, so one might as well consider the punctured code where such a position has been removed, which will have the same dimension and minimum distance. Decoding a degenerate code can therefore be reduced using puncturing to decoding a nondegenerate code. Note that since the codes $\mathcal{C}_{\mathcal L}(D,\Con(G))$ and $\mathcal{C}_{\mathcal L}(D,G^\prime)$ are monomially equivalent, any (list) decoding algorithm for $\mathcal{C}_{\mathcal L}(D,G^\prime)$ immediately gives a (list) decoding algorithm for $\mathcal{C}_{\mathcal L}(D,\Con(G))$. The added complexity is that of dividing and multiplying with the column multipliers $\tilde{f}(P_i)$, which only costs $\bigO(n)$ operations in $\mathbb{F}_{q^e}$ and hence $\softO(ne)$ operations in $\field$. Moreover, as we will see, we will be able to choose $e$ so small that it will not affect the decoding complexity at all in the $\softO$ notation.

Now we state the simplifying assumptions and notation that will be used in the remainder of this article. 

\begin{enumerate}
\item We assume that $G$ is an effective divisor, whose degree satisfies equation \eqref{eq:degree_bounds_G}.
\item We assume that apart from the rational places in $D=P_1+\cdots+P_n$, the function field $\ffield$ has at least one more rational place $\Pinf$. The place $\Pinf$ may or may not be in $\supp(G)$.
\item There exists a rational place $P_0$ of $\ffield$ which has $x$ as a local parameter, where $x \in \ffield$ is a function with pole at $\Pinf$ only of minimal pole order $\mu$. The place $P_0$ may be in $\supp(D)$, but is not in $\supp(G)$.
\end{enumerate}

Let us quickly assess the size of the needed extension degree $e$ in order to satisfy all three item simultaneously. Although one likely can do better, for our purposes it is sufficient to pick $e=e_1e_2e_3$, where $e_1,e_2,e_3$ are given below: to satisfy the first item, we extend $\field$ to $\mathbb{F}_{q^{e_1}}$, where $e_1  = 1+\lceil\log_q(n)\rceil$, using Lemma \ref{lem:G_effective}. To satisfy the second item, we apply Lemma \ref{lem:more_rational_places} with $N=n+1$. Hence we can choose $e_2=\lceil2\log_{q^{e_1}} \max\{n+1,2g + 1\}\rceil$ and extend $\mathbb{F}_{q^{e_1}}$ to $\mathbb{F}_{q^{e_1e_2}}$. For the third item, we need apart from $\Pinf$ and possible rational places in $\supp(G)$, an additional $3g+1$ rational places. Since $G$ is effective, we can apply Lemma \ref{lem:more_rational_places} with $N=1+\deg(G)+3g+1$, so using equation \eqref{eq:degree_bounds_G}, we can choose $e_3=\lceil 2\log_{q^{e_1e_2}}(5g + 1+n)\rceil$ extending $\mathbb{F}_{q^{e_1e_2}}$ to $\mathbb{F}_{q^{e}}$.
Using that $\log_{q^f}(A)=\log_q(A)/f$, it is easy to see that
\begin{align*}
e=e_1e_2e_3 & \le 2\log_{q}(5g + 1+n)+e_1e_2 \\
 & \le 2\log_{q}(5g + 1+n)+2\log_{q} \max\{n+1,2g + 1\}+e_1\\
 & \le 2\log_{q}(5g + 1+n)+2\log_{q} \max\{n+1,2g + 1\}+\log_q(n)+2.
\end{align*}
Hence the overall conclusion is that in terms of complexity only a logarithmic factor in $n+g$ is introduced when reducing from the general case to the simpler setting. In the remainder of this article, instead of writing $\mathbb{F}_{q^{e}}$, we will simply write $\field$ and assume $q$ is large enough so that all three simplifying assumptions stated above are satisfied.

\subsection{Shifted Popov forms of polynomial matrices}
Our decoder relies on efficient algorithms for (free) $\field[x]$-submodules of $\field[x]^m$; in the current subsection, we present well known results and definitions that we need needed for our use cases. For a comprehensive introduction, the reader is referred to \cite{vincent_neiger_bases_2016} and the references within.

We begin with a definition which, among other things, allows us to measure ``size'' of elements in $\field[x]^m$.

\begin{definition}
  \label{def:shifted-deg}
  For any polynomial vector $\vec{v} = (v_1,\dots,v_m) \in \field[x]^m$ and any $\vec{s} = (s_1,\dots,s_m) \in \ZZ^m$ (which we refer to as a \emph{shift}), we define the \emph{$\vec{s}$-degree} of $\vec{v}$ as
  \[
    \deg_{\vec{s}} \vec{v} = \max_k \{ \deg v_k + s_k \} \ .
  \]
  Furthermore, if $k \in \{1,\dots,m\}$ is maximal such that $\deg v_k + s_k = \deg_{\vec{s}} \vec{v}$, then we say that $v_k$ is the \emph{$\vec{s}$-pivot} of $\vec{v}$, and $k$ is its \emph{$\vec{s}$-pivot index}. If $\vec{s} = 0$, then we might omit writing $\vec{s}$ in the above notation, i.e. we might simply write: \emph{pivot}, \emph{pivot index} and \emph{degree}, denoting the latter by $\deg \vec{v} := \deg_{\vec{0}} \vec{v}$.pp
\end{definition}

Any $\field[x]$-basis of a submodule $\V \subseteq \field[x]^m$ of rank $m$ can be described using a nonsingular polynomial matrix $\mat{V} \in \field[x]^{m \times m}$ by identifying the basis elements with the rows of $\mat{V}$. This way, $\V$ is viewed as the $\field[x]$-row space of $\mat{V}$. We will be interested in obtaining the basis whose elements are ``smallest possible''; the following definition makes this notion precise in the context of polynomial matrices.

\begin{definition}
  Given a shift $\vec{s} \in \ZZ^{m}$, a nonsingular matrix $\mat{P} \in \field[x]^{m \times m}$ is said to be in \emph{$\vec{s}$-Popov form} if all of the $\vec{s}$-pivots of its rows lie on the diagonal, are monic and have degrees strictly greater than all other entries in their respective columns. Furthermore, if $\mat{P}$ shares its $\field[x]$-row space with some matrix $\mat{V} \in \field[x]^{r \times m}$, where $m \le r$, then $\mat{P}$ is said to be \emph{the $\vec{s}$-Popov form} of $\mat{V}$.
\end{definition}

Below, we summarize a few important structural properties of shifted Popov forms.

\begin{proposition}[{\cite[Section 1.1]{vincent_neiger_bases_2016}}]
  \label{prop:popov-properties}
  For any nonsingular matrix $\mat{V} \in \field[x]^{m \times m}$ and any shift $\vec{s} \in \ZZ^m$, there exists a unique matrix $\mat{P} \in \field[x]^{m \times m}$ in $\vec{s}$-Popov form having the same $\field[x]$-row space as $\mat{V}$. Furthermore, $\mat{P}$ has minimal shifted row degrees in the following sense: for any $\mat{V} \in \field[x]^{m \times m}$ with the same row space as $\mat{P}$, there exists a bijection between the rows of the two matrices such that the $\vec{s}$-degree of any row of $\mat{V}$ is no smaller than that of the corresponding row of $\mat{P}$. Finally, for any nonzero vector $\vec{v} \in \field[x]^{1 \times m}$ in the row space of $\mat{P}$ with $\vec{s}$-pivot index $k$ it holds that $\deg_{\vec{s}} \vec{v} \ge \deg_{\vec{s}} \vec{p}^{(k)}$, where $\vec{p}^{(k)}$ denotes the $k$-th row of $\mat{P}$.
\end{proposition}

We conclude this subsection with a few complexity bounds.

\begin{proposition}[{\cite[Theorem 1.3]{neiger_computing_2017}}]
  \label{prop:popov-cost}
  There is a deterministic algorithm which for any shift $\vec{s} \in \ZZ^m$ computes the $\vec{s}$-Popov form of any nonsingular matrix $\mat{V} \in \field[x]^{m \times m}$ using $\softO(m^\omega\deg \mat{M})$ operations in $\field$, where $\deg \mat{M}$ denotes the maximal degree among all entries in $\mat{M}$.
\end{proposition}

\begin{proposition}[\cite{zhou_computing_2013}]
  \label{prop:row-basis-cost}
  There is a deterministic algorithm which for any matrix $\mat{V} \in \field[x]^{r \times m}$ with $m \le r$ computes an $\field[x]$-basis of the row space of $\vec{V}$ using $\softO(rm^{\omega - 1} \deg \mat{V})$ operations in $\field$.
\end{proposition}

Combining \cref{prop:row-basis-cost} with \cref{prop:popov-cost}, we obtain the following:

\begin{corollary}
  \label{cor:rect-popov-cost}
  For any shift $\vec{s} \in \ZZ^m$ and any matrix $\mat{V} \in \field[x]^{r \times m}$ with rank $m \le r$, we can compute the $\vec{s}$-Popov form of $\mat{V}$ using $\softO(rm^{\omega - 1}\deg \mat{V})$ operations in $\field$.
\end{corollary}

\section{Representation of function field elements}\label{sec:representation}

For any divisor $A$ of $\ffield$, let $\Ya(A) = \bigcup_{m=-\infty}^{\infty}\L(m\Pinf + A)$ and let $\Ya = \Ya(0)$. Note that $\Ya$ is a ring and $\Ya(A)$ a $\Ya$-module. In fact more can be said: $\Ya$ is a Dedekind domain and $\Ya(A)$ is a fractional ideal of $\Ya$, \cite[Section 1.2]{NX01}.

Modules of the form $\Ya(A)$ are essentially already considered for decoding in \cite{KP95}, also see \cite{beelen_decoding_2008,lee_unique_2014}.
As in \cite{lee_unique_2014}, for any nonzero $a \in \Ya(A)$ we denote by $\delta_A (a)$ the smallest integer $m$ such that $a \in \L(m\Pinf + A)$, i.e. $\delta_A(a) = -\val(a) - \val(A)$ and let $\delta(a) = \delta_0(a) = -\val(a)$. We will take as convention that $\delta_A(0) = - \infty$. Note that for any $a \in \Ya(A)$ and $b \in \Ya(B)$, one has $\delta_{A+B}(ab)=\delta_A(a)+\delta_B(b).$

It is well known that any fractional ideal of a Dedekind domain can be generated by at most two elements \cite[Corollary 2 to Theorem 4]{frohlich-taylor1991book}, but for our purposes we need to know some properties of these generators.
\begin{lemma}\label{lem:YaA-over-Ya}
Let $A=\sum_{i=1}^t n_iQ_i$ be a divisor of $\ffield$ and write $\frak a= \sum_{i} \deg Q_i$.
Then $\Ya(A)$ can be generated as a $\Ya$-module by two elements $a_1$ and $a_2$ satisfying $\delta_A(a_1) \le 2g-1-\deg(A)+\mathfrak a$ and $\delta_A(a_2) \le 4g-2-\deg(A)+\mathfrak a$.
\end{lemma}
\begin{proof}
Prime ideals of $\Ya$ correspond exactly to places of $F$ distinct from $\Pinf$. Therefore, from the proof of Corollary 2 to Theorem 4 in \cite{frohlich-taylor1991book}, we see that two elements $a_1,a_2 \in \Ya(A)$ generate $\Ya(A)$ as $\Ya$-module if and only if for all places $Q_i \in \supp(A)$ distinct from $\Pinf$, we have $\min\{v_{Q_i}(a_1),v_{Q_i}(a_2)\}=-n_i$ and for any other place $Q \neq \Pinf$ of $F$ we have $\min\{v_{Q}(a_1),v_{Q}(a_2)\}=0.$ We will construct two such elements.

Write $m_1=2g-1-\deg(A)+\frak a$. For $j=1,\dots,t$, choose $a_1^{(j)} \in \L(A-\sum_{i \neq j} Q_i+m_1\Pinf) \setminus \L(A-\sum_{i} Q_i+m_1\Pinf)$. Note that such $a_1^{(j)}$ exist, since by the Riemann-Roch theorem, $\dimL(A-\sum_{i \neq j} Q_i+m_1\Pinf)>\dimL(A-\sum_{i} Q_i+m_1\Pinf)$. Defining $a_1=\sum_{j=1}^t a_1^{(j)}$, we see that $v_{Q_i}(a_1)=-n_i$ for $j=1,\dots,t$, while $v_Q(a_1) \ge 0$ for any other place $Q$ distinct from $\Pinf$. In particular $a_1 \in \L(A+m_1\Pinf)$, whence $\delta_A(a_1) \le m_1.$

Now suppose that $Q_{t+1},\dots,Q_{t+s}$ are the zeroes of $a_1$ not in $\supp(A) \cup \{\Pinf\}$. Since $a_1 \in \L(A+m_1\Pinf)$, we see that $\sum_{i=t+1}^{t+s} \deg(Q_i) \le \deg(A)+m_1$. Now define $m_2=2g-1+m_1=4g-2-\deg(A)+\frak a$. Similarly as above, we can construct $a_2 \in \L(A+m_2\Pinf)$, such that $v_{Q_i}(a_2)=0$ for $i=t+1,\dots,t+s$. By construction $\delta_A(a_2) \le m_2$. For $i=1,\dots,t$, we have $v_{Q_i}(a_2) \ge -n_i$ and $v_{Q_i}(a_1)=-n_i$, whence $\min\{v_{Q_i}(a_1),v_{Q_i}(a_2)\}=-n_i$. If $Q \not \in \supp(A) \cup \{\Pinf\}$ is not a zero of $a_1$, then $\min\{v_{Q}(a_1),v_{Q}(a_2)\}=0$, since $v_Q(a_2) \ge 0$. If $Q \not \in \supp(A) \cup \{\Pinf\}$ is a zero of $a_1$, then $v_Q(a_2)=0$, so that also in this case $\min\{v_{Q}(a_1),v_{Q}(a_2)\}=0$. Hence $a_1$ and $a_2$ as constructed above, generated $\Ya(A)$ as a $\Ya$-module.
\end{proof}

As $x \in \Ya \setminus \field$, we can also view $\Ya(A)$ as a free $\field[x]$-module. Following \cite{lee_unique_2014}, we consider a special set of generators of $\Ya(A)$ as $\field[x]$-module, which they called the Ap\'ery system of $\Ya(A)$.

\begin{definition}\label{def:yiA}
  For any divisor $A$ let $\y[A]_i \in \A_i$ be such that $\delta_A(\y[A]_i) \leq \delta_A(a)$ for all $a \in \A_i$, where $i = 0,\dots,\mu-1$ and
  \[
    \A_i = \{ a \in \Ya(A) \mid \delta_A(a) \equiv i \mod \mu \} \ .
  \]
  We also define $\y_i = \y[0]_i$.
\end{definition}

\begin{lemma}
  \label{lem:Fx_basis}
  For any divisor $A$ it holds that
  \begin{enumerate}
  \item $\y[A]_0,\dots,\y[A]_{\mu-1}$ is an $\field[x]$-basis of $\Ya(A)$ and
  \item $-\deg A \leq \delta_A(\y[A]_i) \leq 2g - 1 - \deg(A) + \mu$ for $i = 0,\dots,\mu-1$.
  \end{enumerate}
\end{lemma}
\begin{proof}
The first statement is from \cite{lee_unique_2014}. For the convenience of the reader we give a proof. From the strict triangle inequality for $v_{\Pinf}$, it is clear that the elements $\y[A]_0,\dots,\y[A]_{\mu-1}$ are linearly independent over $\field[x]$. Also, it is clear that $\Y \subseteq \Ya(A)$, where $\Y = \langle \y[A]_0,\dots,\y[A]_{\mu-1} \rangle_{\field[x]}$. If $\Y \neq \Ya(A)$, then there would exist $a \in \Ya(A) \setminus \Y$, such that $\delta_A(a) > -\infty$ is minimal. Write $\delta_A(a) = m\mu + r$ and $\delta_A(\y[A]_{r}) = m'\mu + r$, where $m,m',r \in \ZZ$ with $0 \leq r < \mu$. Note that $m' \leq m$ by definition of $\y[A]_{r}$.
Since
  \[
    \delta_A(x^{m-m'}\y[A]_{r}) = \delta(x^{m-m'}) + \delta_A(\y[A]_{r}) = (m-m')\mu + (m'\mu + r) = \delta_A(a) \ ,
  \]
there exists a constant $\beta \in \field$ such that $\delta_A(c) < \delta_A(a)$, where $c = a - \beta x^{m-m'}\y[A]_{r} \in \Ya(A)$.
The minimality of $\delta_A(a)$ guarantees that $c \in \Y$, however, this would imply that $a = c + \beta x^{m - m'}\y[A]_{r} \in \Y$. Hence $\Y = \Ya(A)$, which is a contradiction.

In the second statement, the lower bound simply follows from the fact that $\y[A]_i \in \L(\delta_A(\y[A]_i)\Pinf + A) \neq \{0\}$. For the upper bound it is sufficient to show that for every integer $m > 2g-1 -\deg(A)$ there exists an $a \in \Ya(A)$ with $\delta_A(a) = m$.
  But indeed, if $m > 2g-1 - \deg(A)$, then $\deg(m\Pinf + A) > 2g - 1$, and so \cite[Theorem 1.5.17]{stichtenoth_algebraic_2009} implies that
  \[
    \L(m\Pinf + A)
    \neq
    \L \big(
    (m-1)\Pinf + A
    \big) \ ,
  \]
  which concludes the proof.
\end{proof}
For later use, we also state the following lemma.
\begin{lemma}
  \label{lem:deg}
  If $a = \sum_{i=0}^{\mu-1} a_i\y[A]_i \in \Ya(A)$, where $a_i \in \field[x]$ and $A$ is a divisor, then
  \[
    \deg a_i \leq \frac{1}{\mu}(\delta_A(a) - \delta_A(\y[A]_i)) \leq \frac{1}{\mu}(\delta_A(a) + \deg A) \ .
  \]
\end{lemma}
\begin{proof}
  Simply observe that for $i = 0,\dots,\mu-1$ it holds that
  \[
    \delta_A(a)
    = \max_j \delta_A(a_j\y[A]_j)
    \geq \delta(a_i) + \delta_A(\y[A]_i)
    \geq \delta(a_i) - \deg A \ ,
  \]
  where the equality follows from the strict triangle inequality for $v_{\Pinf}$ and second inequality is given by \cref{lem:Fx_basis}.
  But then
  \[
    \deg a_i = \delta(a_i)/\mu \leq \frac{1}{\mu}(\delta_A(a) - \delta_A(\y[A]_i)) \leq \frac{1}{\mu}(\delta_A(a) + \deg A) \ .
  \]
\end{proof}

\section{Guruswami-Sudan Decoding}\label{sec:Guruswami-Sudan}
In this section, we paraphrase the Guruswami-Sudan list decoding algorithm \cite{guruswami_improved_1999} for $\code$ and formulate it in terms of $\Ya$ modules. For the remainder of this paper fix $s,\ell \in \ZZ_{>0}$, $s \le \ell$, where $s$ is the multiplicity parameter and $\ell$ the designed list size of the Guruswami-Sudan list decoder. The corresponding list decoding radius will be denoted by $\tau$.

\begin{definition}
  Let $P$ be a rational place of $\ffield$, $r \in \field$ and $Q \in \ffield[z]$. We will say that ``$Q$ has a root of multiplicity $s$ at $(P,r)$'' if for any local parameter $\phi$ of $P$, there exist $c_{a,b} \in \field$ such that
  \[
    Q = \sum_{\substack{a,b \geq 0 \\ a + b \geq s}}c_{a,b}\phi^a(z-r)^b
  \]
  with $c_{a,s-a} \neq 0$ for at least one $0 \leq a \leq s$.
\end{definition}
A consequence of this definition is the following:

\begin{lemma}
  \label{lem:multzero_at_fun}
  If $Q \in \ffield[z]$ has a root of multiplicity $s$ at $(P, r)$ and $f \in \ffield$ is such that $f(P) = r$, then
  $\val[P](Q(f)) \geq s$.
\end{lemma}
\begin{proof}
  Writing
  \[
    Q(f) = \sum_{\substack{a,b \geq 0 \\ a + b \geq s}} c_{a,b}\phi^a(f - r)^b \ ,
  \]
  where $\phi$ is any local parameter of $P$ and $c_{a,b} \in \field$, the triangle inequality directly implies that
  \[
    \val[P](Q(f))
    \geq \min_{\substack{a,b \geq 0 \\ a + b \geq s}} \big( \val[P](\phi^a) + \val[P]((f-r)^b) \big)
    \geq \min_{\substack{a,b \geq 0 \\ a + b \geq s}}(a + b) \geq s \ .
  \]
\end{proof}

For any $Q = \sum_{t = 0}^\ell z^tQ^{(t)}$ with $Q^{(t)} \in \Ya(-tG)$ we define $\delta_G(Q) = \max_t \delta_{-tG}(Q^{(t)})$. Moreover, for a given received word $\r=(r_1,\dots,r_n) \in \field^n$, we write

\begin{equation}\label{def:M}
\begin{multlined}
\M=\{ Q = \sum_{t = 0}^{\ell}z^tQ^{(t)} \in \ffield[z] \mid  Q^{(t)} \in \Ya(-tG),  \\
 \text{$Q$ has a root of multiplicity at least $s$ at $(P_j,r_j)$ for $j=1,\dots,n$}\}.
\end{multlined}
\end{equation}
\begin{theorem}[Special case of Guruswami--Sudan \cite{guruswami_improved_1999}]
Let $\r$ be a received word and $Q \in \M$ with $\delta_G(Q) < s(n-\tau)$. If $f \in \L(G)$ such that the Hamming weight of $\r-\ev_D(f)$ is at most $\tau$, then $Q(f)=0$.
\end{theorem}

\begin{proof}
  Since $f^t \in \L(tG)$ and $Q^{(t)} \in \Ya(-tG)$, then $f^tQ^{(t)} \in \Ya$, and consequently $Q(f) \in \Ya$.
  Furthermore, since $\delta_{tG}(f^t) \le 0$, then by the triangle inequality
  \[
    \delta(Q(f)) \leq \max_t \delta_{-tG}(Q^{(t)}) = \delta_G(Q) \ .
  \]
 We write $\errors=\{j \mid r_j\neq f(P_j)\}$. Note that the cardinality of $\errors$ is at most $\tau$. Since $f(P_j) = r_j$ for $j \not \in \errors$, it follows from \cref{lem:multzero_at_fun} that $Q(f) \in \Ya(-T)$, where $T = s\sum_{j \not \in \errors}P_j$.
  Since $\delta_G(Q) < s(n-\tau) \leq \deg T$, we may conclude that
  \[
    Q(f) \in \L \big(\delta_G(Q)\Pinf - T \big) = \{0\} \ .
  \]
\end{proof}

\subsection{Structure of $\M$ as a $\Ya$-module}
\label{subsec:structure-of-M-as-Ya-module}

The set $\M$ introduced in equation \eqref{def:M} is easily seen to be a module over the ring $\Ya$. In this subsection, we determine some of its structural properties. For the remainder of this article let $G_t = -tG - \max\{0,s-t\}D$ for $t = 0,\dots,\ell$.

\begin{theorem}
  \label{thm:M-description}
Let $\vec{r}=(r_1,\dots,r_n)$ be a received word and $R \in \Ya(G)$ be such that $R(P_j) = r_j$ for $j = 1,\dots,n$. Then it holds that
  \[
    \M = \bigoplus_{t=0}^{\ell}(z-R)^t \Ya(G_t) \ .
  \]
\end{theorem}
\begin{proof}
Note that for all $j$ and all $h \in \Ya(G_t)$, $v_{P_j}(h) \ge \max\{0,s-t\}$. Further $(z-R)^t$ has a root of multiplicity $t$ at $(P_j,r_j)$, since $R(P_j)=r_j$. Hence any element in $(z-R)^t \Ya(G_t)$ has a root of multiplicity at least $s$ at $(P_j,r_j)$.
Moreover, since $R \in \Ya(G)$, we see that $(z-R)^t \Ya(G_t)=\left(\sum_{u=0}^t z^u\binom{t}{u}(-R)^{t-u}\right)\Ya(G_t) \subseteq \bigoplus_{u=0}^tz^u \Ya(G_u)$. Hence $(z-R)^t \Ya(G_t) \subseteq \M$. Since $\M$ is a $\Ya$-module, this implies that $\bigoplus_{t=0}^{\ell}(z-R)^t \Ya(G_t) \subseteq  \M.$

We will prove the reverse inclusion $\M \subseteq  \bigoplus_{t=0}^{\ell}(z-R)^t \Ya(G_t)$ by induction on $s$.
Let $Q=\sum_{t = 0}^{\ell}z^tQ^{(t)} \in \M$ and write $Q=\sum_{t = 0}^{\ell}(z-R)^t\tilde{Q}^{(t)}$ for certain $\tilde{Q}^{(t)} \in F$. Writing $z^t=((z-R)+R)^t$ and using Newton's binomium, we obtain
\[\tilde{Q}^{(t)}=\sum_{u=t}^\ell \binom{u}{t}R^{u-t} Q^{(u)} \in \Ya(-tG), \text{ for $t=0,\dots,\ell$,}\]
since $R \in \Ya(G).$ Now observe that \cref{lem:multzero_at_fun} implies that $\tilde{Q}^{(0)}=Q(R) \in \Ya(G_0).$

Now if we assume $s=1$, then $\Ya(G_t)=\Ya(-tG)$ for $t>0$ and we can conclude from the above that $Q \in \bigoplus_{t=0}^{\ell}(z-R)^t \Ya(G_t)$.

If $s>1$, we proceed as follows: using $\tilde{Q}^{(0)} \in \Ya(G_0) \subseteq \M$, we conclude
\[\M \ni Q-\tilde{Q}^{(0)}=(z-R) \cdot \sum_{t = 0}^{\ell-1}(z-R)^t\tilde{Q}^{(t+1)} \, .\]
Since $z-R$ has a root of multiplicity one at $(P_j,r_j)$ for all $j$, we see that $\sum_{t = 0}^{\ell-1}(z-R)^t\tilde{Q}^{(t+1)}$ has a root of multiplicity at least $s-1$ at $(P_j,r_j)$ for all $j$. Hence $\sum_{t = 0}^{\ell-1}(z-R)^t\tilde{Q}^{(t+1)} \in \mathcal M_{s-1,\ell}(D,G)$. Then using the induction hypothesis for $s-1$, we may conclude that  $Q \in \bigoplus_{t=0}^{\ell}(z-R)^t \Ya(G_t).$
\end{proof}

\begin{corollary}[of \cref{thm:M-description} and \cref{lem:YaA-over-Ya}]
  \label{cor:Ya-basis-of-M}
  It holds that
  \begin{align*}
    \M &= \langle B_v^{(u)} | u=0,\dots,\ell, \ v = 1,2 \rangle_{\Ya} \ , \quad \text{where} \\
    B_v^{(u)} &= (z-R)^u\g{u}_v = \sum_{r=0}^u \binom{u}{r}z^r(-R)^{u-r} \g{u}_v
                \in \bigoplus_{t=0}^{\ell}z^t\Ya(-tG) \ ,
  \end{align*}
  with $g^{(u)}_1,g^{(u)}_2 \in \Ya(G_u)$ such that $\langle g^{(u)}_1,g^{(u)}_2 \rangle_{\Ya} =  \Ya(G_u)$, $$\delta_{G_u}(g_1^{(u)}) \le 2g-1+(u+1)\deg(G)+\max\{0,s-u+1\}n,$$ and $$\delta_{G_u}(g_2^{(u)}) \le 4g-2+(u+1)\deg(G)+\max\{0,s-u+1\}n.$$
\end{corollary}
\begin{proof}
The first part directly follows from \cref{thm:M-description} and \cref{lem:YaA-over-Ya}. To obtain the stated upper bounds on $\delta_{G_u}(g_1^{(u)})$ and $\delta_{G_u}(g_2^{(u)})$ from \cref{lem:YaA-over-Ya}, note that $\sum_{Q \in \supp(G)} \deg(Q) \le \deg(G)$, since $G$ is an effective divisor. Hence $\sum_{Q \in \supp(G_u)} \deg(Q) \le \deg(G)+n$ if $u<s$, while $\sum_{Q \in \supp(G_u)} \deg(Q) \le \deg(G)$ if $u\ge s.$ The stated upper bounds are implied by this.
\end{proof}
Note that the proof of the corollary actually implies that for $u=s$, the stated upper bounds for $\delta_{G_u}(g_1^{(u)})$ and $\delta_{G_u}(g_2^{(u)})$ can be improved by $n$.

For computational purposes, we will later view $\M$ as an $\field[x]$ module. Since any element from $\Ya$ is an $\field[x]$-linear combination of $y_0,\dots,y_{\mu-1}$, we obtain the following.
\begin{corollary}\label{cor:Fqx-basis-of-M}
It holds that $\M = \langle y_iB_v^{(u)} | i=0,\dots,\mu-1,u=0,\dots,\ell, \ v = 1,2 \rangle_{\field[x]}$.
\end{corollary}

\begin{remark}\label{rem:alternative-Fqx-basis-of-M}
Since $G_t=-tG$ for $s \le t \le \ell$, a minor modification of the proof of \cref{thm:M-description} shows that $\M = \bigoplus_{t=0}^{s}(z-R)^t \Ya(G_t) \oplus \bigoplus_{t=s+1}^\ell (z-R)^sz^{t-s} \Ya(G_t).$ This shows that the elements $\tilde{B}_v^{(u)} = B_v^{(u)}$ if $u \le s$ together with $\tilde{B}_v^{(u)} (z-R)^sz^{t-s}\g{u}_v$ if $s<u\le \ell$ form an alternative set of generators over $\Ya$ for $\M$. Likewise the elements in the set $\{y_i\tilde{B}_v^{(u)} | i=0,\dots,\mu-1,u=0,\dots,\ell, \ v = 1,2\}$ generate $\M$ as an $\field[x]$-module. If $s < \ell$, these alternative generators can be computed using fewer operations and are therefore in general preferable.
\end{remark}

\begin{remark}
The $\Ya$-module $\M$ is an example of a torsion free, finitely generated module of rank $\ell+1$. Though we will not need this in the following, it interesting to note that any torsion free, finitely generated module $\mathcal M$ of rank $r$ over a Dedekind domain $\Ya$, is isomorphic to a direct product of $r$ fractional ideals of $\Ya$, say $\mathcal M \cong I_1 \oplus \cdots \oplus I_{r}$. Moreover, the product $I=I_1\cdots I_r$ of these fractional ideals modulo principal fractional ideals only depends on the isomorphism class of $\mathcal M$. Therefore the element of the ideal class group of $\Ya$ corresponding to $I$ is called the Steinitz invariant of $\mathcal M$. See \cite[Section II.4]{frohlich-taylor1991book} for more details. \cref{thm:M-description} can be reformulated as $\M \cong \bigoplus_{u=0}^\ell \Ya(G_u)$ and in particular the Steinitz invariant of $\M$ is the element in the ideal class group of $\Ya$ corresponding to
\[\prod_{u=0}^\ell \Ya(G_u)=\Ya(\sum_{u=0}^\ell G_u)=\Ya(-\binom{\ell+1}{2}G-\binom{s+1}{2}D).\]
\end{remark}

Returning to decoding, given a received word $\r$, code $\code$, and parameters $s,\ell$, the main steps in our algorithmic approach to Guruswami-Sudan list decoding are the following.
\begin{enumerate}
\item Compute a generating set over $\Ya$ of $\M$. We will do this in Subsection \ref{subsec:genYa}
\item Compute a generating set over $\field[x]$ of $\M$. We will address this in Subsection \ref{subsec:genFqx}
\item Using fast row reduction over $\field[x]$, find a nonzero $Q \in \M$ satisfying $\delta_G(Q)<s(n-\tau)$. See Subsection \ref{subsec:findQ}
\item Find the roots of $Q$ in $\L(G)$. See Subsection \ref{subsec:findroots}
\end{enumerate}
As we will see, the main result of this paper is that all these steps can be done in complexity $\softO(\mu^{\omega-1}\ell^{\omega+1}(n+g))$ and with a slight variation even in $\softO(\mu^{\omega-1}s\ell^{\omega}(n+g))$.

To simplify the description of the algorithms in the next sections, it will be convenient to assume that apart from $\Pinf$, the function field $F$ contains an additional $Z:=\deg G+\max\{(\ell+1)\deg G+4g+(s+1)n,\deg G+(\ell+3)(2g-1)+(s+1)n+2+\mu\}$ rational places. Even though this will not be the case in general, the same trick as at the end of Section \ref{sec:preliminaries}, will allow us to assume this. More precisely, the function field $F\mathbb{F}_{q^e}$ with $e=\lceil 2 \log_q(\max\{Z,2g+1\}) \rceil$ will contain at least $1+Z$ rational places by \cref{lem:more_rational_places}. Since using equation \eqref{eq:degree_bounds_G}, $e \in \bigO(\log_q(\ell(n+g)))$, this does not interfere with our target complexity and hence does not result in any loss of generality. We will suppress the exponent $e$ from the notation and will from now on write $\field$ for the finite field we work over, but assume that $F$ contains all the rational places that we need to run the algorithms we describe in the next section (specifically: \cref{algo:basisYa} and \cref{algo:basisFqx}).

\section{Algorithms}\label{sec:algorithms}

In this section, we present the algorithms that we will use to execute the Guruswami-Sudan list decoder. We start with discussing multi-point evaluation and interpolation algorithms that will form the backbone of the algorithms discussed later in the section.

\subsection{Multi-Point Evaluation}
\label{subsec:mpe}

When defining $\code$, we used the evaluation map $\ev_D$. We will later need to be able to compute $\ev_D(f)$ fast, meaning we want to be able to evaluate the function $f \in \L(G)$ in the multiple points $P_1,\dots,P_n$ fast. As a matter of fact, since we will need a slightly more general setting later on, we phrase the results in this and the next subsection in terms of a very similar evaluation map, but avoid to use the divisors $D$ and $G$.

\begin{lemma}
  \label{lem:ev}
Let $A$ be a divisor and $E = E_1 + \cdots + E_N$ for distinct rational places $E_1,\dots,E_N$ of $F$ such that $\supp(A) \cap \supp(E)= \emptyset$. Further denote by $\ev_E: \L(A) \to \field^N$ the evaluation map defined by $\ev_E(a)=(a(E_1),\dots,a(E_N)).$ Then
  \begin{enumerate}
  \item $\ev_{E}$ is injective when $\deg A < \deg E$,
  \item $\ev_E$ is surjective when $\deg A \geq \deg E + 2g - 1$.
  \end{enumerate}
\end{lemma}
\begin{proof}
  For the first item, simply observe that the dimension of the kernel of $\ev_E$ is $\dimL(A - E) = 0$, since $\deg (A-E) < 0$.

  For the second item, observe that the dimension of the image of $\ev_E$ is
  \[
    \dimL(A) - \dimL(A - E) = \deg A - g + 1 - (\deg A - \deg E - g + 1) = \deg E \ ,
  \]
since $\deg(A) \ge 2g-1$ and $\deg(A-E) \ge 2g-1$, see \cite[Theorem 1.5.17]{stichtenoth_algebraic_2009}.
\end{proof}

Now we state Algorithm \ref{algo:evaluate}, which computes $\ev_E(a)$ using the representation of function field elements as introduced in Section \ref{sec:representation}.

\begin{algorithm}[H]
  \caption{$\Evaluate(a,E,A,\vec{x},\vec{y})$} \label{algo:evaluate}
  \begin{algorithmic}[1]
    \Input
    \begin{itemize}
    \item Divisors $A$ and $E = E_1 + \cdots + E_N$, where $E_1,\dots,E_N \in \places_{\ffield}\setminus\{\Pinf\}$ are distinct rational places and $\supp(A) \cap \supp(E)=\emptyset$,
    \item a function $a = \sum_{i=0}^{\mu-1} a_i y^{(A)}_i \in \Ya(A)$, where $a_i \in \field[x]$,
    \item evaluations $\vec{x} = (x_j)_{j=1,\dots,N}$, where $x_j = x(E_j) \in \field$,
    \item evaluations $\vec{y} = (y_{i,j})^{i=0,\dots,\mu-1}_{j=1,\dots,N}$, where $y_{i,j} = y^{(A)}_i(E_j) \in \field$.
    \end{itemize}
    \Output
    \begin{itemize}
    \item Evaluations $\ev_E(a) \in \field^N$.
    \end{itemize}
    \For{$i=0,\dots,\mu-1$}
    \State $(a_{i,1},\dots,a_{i,N}) \in \field^N \assign (a_i(x_1),\dots,a_i(x_N))$ \Comment{Univariate MPE}
    \EndFor
    \State \Return $\sum_{i=0}^{\mu-1} (  a_{i,1}y_{i,1},\dots,a_{i,N}y_{i,N} ) \in \field^N$
  \end{algorithmic}
\end{algorithm}

\begin{lemma}
  \cref{algo:evaluate} is correct and costs $\softO(\mu N + \delta_A(a) + \deg A)$ operations in $\field$.

\end{lemma}
\begin{proof}
  Correctness simply follows from the fact that for $j=1,\dots,N$
  \[
    \sum_{i=0}^{\mu-1} a_{i,j}y_{i,j}
    = \sum_{i=0}^{\mu-1} a_i(x(E_j)) y^{(A)}_i(E_j)
    = \sum_{i=0}^{\mu-1} (a_iy^{(A)}_i)(E_j)
    = a(E_j) \ .
  \]
  For complexity, notice that the total cost of the for-loop on Line 1 amounts to that of evaluating each of the univariate polynomials $a_0,\dots,a_{\mu-1} \in \field[x]$ on $N$ points. According to \cref{lem:deg},
  \[
    \deg a_i \leq \frac{1}{\mu}(\delta_A(a) + \deg A) \quad \text{for } i=0,\dots,\mu-1 \ ,
  \]
  hence the total cost of the for-loop is bounded by
  \[
    \softO(\mu(N + \max_i \deg a_i)) \subseteq \softO(\mu N + \delta_A(a) + \deg A) \ .
  \]
  Line $3$ costs $\bigO(\mu N)$, which is subsumed by the cost of the for-loop.
\end{proof}

\subsection{Interpolation}
\label{subsec:interpolation}

In this subsection, we address the interpolation problem. We start with an existence result.

\begin{lemma}
  \label{lem:interp_size}
Let $A$ be a divisor and $E = E_1 + \cdots + E_N$ for distinct rational places $E_1,\dots,E_N$ of $F$ different from $\Pinf$ such that $\supp(A) \cap \supp(E)= \emptyset$.   
For any $(w_1,\dots,w_N) \in \field^N$ there exists an $a \in \Ya(A)$ with
  \[
    \delta_A(a) \leq \deg E+2g-1-\deg A
  \]
  such that $a(E_j) = w_j$ for $j = 1,\dots,N$.
\end{lemma}

\begin{proof}
Letting $A' = (\deg E+2g-1-\deg A)\Pinf + A$ we get that $\deg A' \geq \deg E + 2g - 1$, which according to \cref{lem:ev} implies that the evaluation map $\ev_{E}: \L(A^\prime) \to \field^N$ is surjective.
\end{proof}

\begin{definition}
  If $E = E_1 + \cdots + E_N$, where $E_1,\dots,E_N$ are distinct rational places different from $\Pinf$, and $U_1,\dots,U_\mu$ are effective divisors satisfying
  \begin{enumerate}
  \item $E = U_1 + \dots + U_\mu$,
  \item $\supp U_i \cap \supp U_j = \emptyset$ when $i \neq j$,
  \item $|\deg U_i - \deg U_j| \leq 1$ for all $i,j \in \{1,\dots,\mu\}$,
  \item for any $E_j,E_k \in \supp U_i$ it holds that $x(E_j) = x(E_k) \iff E_j = E_k$,
  \end{enumerate}
  then we will say that $U_1,\dots,U_\mu$ is an \emph{$x$-partition} of $E$.
\end{definition}

\begin{lemma}
  \label{lem:mu-geq-valency}
  If $\S$ is a set of places such that $x(P) = x(P')$ for all $P,P' \in \S$, then $|\S| \leq \mu$.
\end{lemma}
\begin{proof}
  If $\alpha = x(P)$ for every $P \in \S$, then it is easy to see that
  \[
    0 \neq x - \alpha \in \L(\mu\Pinf - \sum_{P \in \S}P) \ .
  \]
  But if $\mu < |\S|$, then the above Riemann-Roch space has dimension zero.
\end{proof}

\begin{lemma}
  \label{lem:mu-partition}
  There exists an $x$-partition of any divisor of the form $E = E_1 + \cdots + E_N$, where $E_1,\dots,E_N$ are distinct rational places different from $\Pinf$.
\end{lemma}

\begin{proof}
 We use induction on $N$. The base case $N=0$ is trivial, so let us consider the induction step. Suppose $U_1,\dots,U_\mu$ is an $x$-partition of $E - E_N$, and let $a,b$ be such that $U_a$ and $U_b$ have minimal degree among the elements of
  \[
    \{U_1,\dots,U_\mu\} \and \{U_i \mid x(E_j) \neq x(E_N) \text{ for all } E_j \in \supp U_i, i = 1,\dots,\mu \}
  \]
  respectively ($b$ exists due to \cref{lem:mu-geq-valency}). If $\deg U_a = \deg U_b$, then an $x$-partition of $E$ can be obtained by replacing $U_b$ with $U_b + E_N$. If on the other hand $\deg U_a < \deg U_b$, then $U_a$ contains a place $\hat E_a$ with $x(\hat E_a) = x(E_N)$ and $U_b$ contains a place $\hat E_b$ such that $x(\hat E_b) \neq x(E_j)$ for all $E_j \in U_a$. But then an $x$-partition of $E$ can be obtained by replacing $U_a$ with $U_a - \hat E_a + \hat E_b + E_N$ and $U_b$ with $U_b - \hat E_b + \hat E_a$.
\end{proof}

\begin{definition}
  \label{def:hermit-pade}
  For any polynomial matrix $\mat{A} \in \field[x]^{\phi \times \theta}$ with columns $\vec{A}_1,\dots,\vec{A}_\theta$ and any polynomial vector $\vec{u} = (u_1,\dots,u_{\theta}) \in \field[x]^{\theta}$ define the $\field[x]$-module
  \[
    \H_{\vec{u}}(\mat{A}) = \{ \vec{v} \in \field[x]^{\phi} \mid \vec{v} \cdot \vec{A}_k \equiv 0 \pmod{u_k} \for k = 1,\dots,\theta\} \ .
  \]
\end{definition}

The following is a direct adaptation of Theorem 1.7 from \cite{rosenkilde-storjohann-2018}. We also refer to \cite{rosenkilde-storjohann-2018} for the definition of the Popov form and the $(-\vec{d})$-Popov form of a matrix. Note that if $u_1\cdots u_\theta \neq 0,$ the rank of $\H_{\vec{u}}(\mat{A})$ is $\phi$, as $u_1\cdots u_\theta \field[x]^\phi \subseteq \H_{\vec{u}}(\mat{A}) \subseteq \field[x]^\phi.$ Note that the problem of computing the shifted Popov basis of $\H_{\vec{u}}(\mat{A})$ has been studied extensively in the literature. Earlier references than \cite{rosenkilde-storjohann-2018} are for example \cite{jeannerod_computing_2017,jeannerod_fast_2016}


\begin{theorem}[{\cite[Theorem 1.7]{rosenkilde-storjohann-2018}}]
  \label{thm:sim-hermit-pade}
   Assume $\phi,\theta \in \ZZ_{\ge 1}$ are integers such that $\phi \ge \theta$. There exists an algorithm which for any $\mat{A} \in \field[x]^{\phi \times \theta}$, $\vec{u} \in (\field[x]\setminus \{0\})^{\theta}$ and $\vec{d} = (d_1,\dots,d_\phi) \in \ZZ_{\geq 0}^{\phi}$ can compute a matrix $\mat{V} \in \field[x]^{\phi \times \phi}$ in $(- \vec{d})$-Popov form, whose rows form an $\field[x]$-basis of $\H_{\vec{u}}(\mat{A})$. Furthermore, if there exists a vector $\vec{v} = (v_1,\dots,v_\phi) \in \H_{\vec{u}}(\mat{A})$ satisfying the degree constraints $\deg v_t < d_t$ for $t = 1,\dots,\phi$, then at least one row of $\mat{V}$ will also satisfy these constraints. The complexity of such an algorithm can be taken to be $\softO(\phi^{\omega-1}\theta d)$ operations in $\field$, where $d = \max_t d_t + \max_k \deg u_k$.
\end{theorem}

For our purposes, we will sometimes need to allow non-integer shifts $d_1,\dots,d_\phi$. Non-integer, rational shifts were handled in \cite{nielsen_sub-quadratic_2015} essentially by permuting columns in a very specific way:
\begin{theorem}[Reformulation of Corollary 12 in \cite{nielsen_sub-quadratic_2015}]\label{thm:permute}
  Let $\mat{V} \in \field[x]^{\gamma \times \phi}$ and $\vec{d} = (d_1/\mu, \dots,d_\phi/\mu) \in (\frac{1}{\mu}\ZZ)^{\phi} \subset (\frac{1}{\mu}\ZZ)^{\phi}$, where $d_1,\dots,d_\phi \in \ZZ$. If $\pi$ is the permutation on $\{1,\dots,\phi\}$ defined by
  \begin{align*}
    \pi(i) > \pi(j) \iff
    \begin{array}{c}
        (d_i \rem \mu) > (d_j \rem \mu) \\
        \text{or} \\
      (d_i \rem \mu) = (d_j \rem \mu) \and i > j
      \end{array} \ ,
  \end{align*}
  and $\Psi : \field[x]^\phi \to \field[x]^\phi$ is the map
  \[
    (v_1,\dots,v_\phi) \mapsto (x^{\lfloor d_{\pi(1)}/\mu \rfloor}v_{\pi(1)},\dots,x^{\lfloor d_{\pi(\phi)}/\mu \rfloor}v_{\pi(\phi)}) \ ,
  \]
  then $\mat{V}$ is in $\vec{d}$-Popov form if and only if $\Psi(\mat{V})$ is in Popov form, where $\Psi(\mat{V}) \in \field[x]^{\gamma \times \phi}$ is the matrix created by applying $\Psi$ to each row of $\mat{V}$.
\end{theorem}

Using the permutation defined in \cref{thm:permute} in combination with \cref{thm:sim-hermit-pade}, we obtain the following:
\begin{corollary}\label{cor:sim-hermit-pade}
  In the context of \cref{thm:sim-hermit-pade} we can allow $\vec{d} \in (\frac{1}{\mu}\ZZ)^{\phi}$ and find the desired matrix $\mat{V} \in \field[x]^{\phi \times \phi}$ in complexity $\softO(\phi^{\omega-1}\theta d)$ operations in $\field$, where $d = \max_t |d_t| + \max_k \deg u_k$.
\end{corollary}
\begin{proof}
  Write $\vec{d} = (\tilde{d}_1/\mu,\dots,\tilde{d}_\phi/\mu)$ with $\tilde{d}_t \in \ZZ$ and notice that \cref{thm:permute} implies that $\mat{V} \in \field[x]^{\phi \times \phi}$ is in $(-\vec{d})$-Popov form if and only if $\tilde{\mat{V}} \in \field[x]^{\phi \times \phi}$ is in $(-\tilde{\vec{d}})$-Popov form, where
  \[
    \tilde{\vec{d}} = (\lfloor \tilde{d}_{\pi(1)}/\mu \rfloor, \dots, \lfloor \tilde{d}_{\pi(\phi)}/\mu \rfloor) \in \ZZ^\phi \ ,
   \]
   and $\tilde{\mat{V}}$ is matrix obtained from $\mat{V}$ by permuting its columns using $\pi$ from \cref{thm:permute}. By \cref{thm:sim-hermit-pade}, for any matrix $\mat{A} \in \field[x]^{\phi \times \theta}$, we can compute the basis $\tilde{\mat{V}} \in \field[x]^{\phi \times \phi}$ of $\H_{\vec{u}}(\tilde{\mat{A}})$ in $(-\tilde{\vec{d}})$-Popov form, where $\tilde{\mat{A}} \in \field[x]^{\phi \times \theta}$ is obtained by permuting the rows of $\mat{A}$ by $\pi$, as long as the entries of $\tilde{\vec{d}}$ are non-negative. By simply adding the constant $\max_t \lceil |\tilde{d}_t|/\mu \rceil$ to all coordinates of $\tilde{\vec{d}}$, we can ensure that this is true without breaking the target complexity. Finally, it is trivial to obtain $\mat{V}$ from $\tilde{\mat{V}}$ by applying $\pi^{-1}$ to its columns.
\end{proof}

With these algorithmic aspects in place, we turn our attention again to the interpolation problem. We start with a lemma, which will give rise to our interpolation algorithm directly.

\begin{lemma}
  \label{lem:interp-correct}
Let $A$ be a divisor and $E = E_1 + \cdots + E_N$ for distinct rational places $E_1,\dots,E_N$ of $F$ different from $\Pinf$ such that $\supp(A) \cap \supp(E)= \emptyset$. Let $(w_1,\dots,w_N) \in \field^N$ as well as an $x$-partition $U_1,\dots,U_\mu$ of $E$ be given.

Suppose that $\mat{T} = [T_k] \in \field[x]^{1 \times \mu}$ and $\mat{S} = [S_{i,k}] \in \field[x]^{\mu \times \mu}$ are such that $$T_k(x(E_j)) = -w_j \text{ for all } E_j \in \supp(U_k)$$ and $$S_{i,k}(x(E_j)) = \y[A]_i(E_j) \text{ for all } E_j \in \supp(U_k).$$

If $\vec{u} = (u_1,\dots,u_\mu) \in \field[x]^\mu$, where $u_k = \prod_{E_j \in \supp(U_k)}(x - x(E_j))$, and $\vec{d} = (d_0,\dots,d_{\mu-1},0) \in (\frac{1}{\mu}\ZZ)^{\mu+1}$, where
  \[
    d_i = \frac{1}{\mu} (\deg E + 2g - \deg A - \delta_A(\y[A]_i)) \for i = 0,\dots,\mu-1 \ ,
  \]
  then in the $(-\vec{d})$-Popov basis of $\H_{\vec{u}}(\mat{A})$, where
  \[
    \mat{A} =
    \left[
      \begin{array}{c}
        \mat{S} \\
        \hline
        \mat{T}
      \end{array}
    \right] \in \field^{(\mu+1) \times \mu} \ ,
  \]
  there exists a vector $\vec{a} = (a_0,\dots,a_{\mu-1},1) \in \field[x]^{\mu+1}$ with $\deg a_i < d_i$ for $i = 0,\dots,\mu-1$.
  Moreover, if
  \[
    a = \sum_{i=0}^{\mu-1} a_i \y[A]_i \ ,
  \]
  then $\delta_A(a) \leq \deg E + 2g - 1 - \deg A$ and $a(E_j) = w_j$ for $j = 1,\dots,N$.
\end{lemma}

\begin{proof}
  Observe that according to \cref{lem:interp_size} there exists a $b \in \Ya(A)$ with
  \[
    \delta_A(b) \leq \deg E + 2g - 1 - \deg A
  \]
  such that $b(E_j) = w_j$ for $j = 1,\dots,N$. If we write $b = \sum_{i=0}^{\mu-1} b_i \y[A]_i$, where $b_i \in \field[x]$, then it follows from \cref{lem:deg} that
  \[
    \deg b_i \leq \frac{1}{\mu}(\delta_A(a) - \delta_A(\y[A]_i)) = \frac{1}{\mu}(\deg E + 2g - 1 - \deg A - \delta_A(\y[A]_i)) < d_i \ .
  \]
We claim that $\vec{b} := (b_0,\dots,b_{\mu-1},1) \in \H_{\vec{u}}(\mat{A})$.
To see this let $c_k = \sum_{i=0}^{\mu-1} b_i S_{i,k} + T_k \in \field[x]$ for $k = 1,\dots,\mu$ and observe that for any $E_j \in U_k$ it holds that
  \[
    c_k(x(E_j)) = \sum_{i=0}^{\mu-1} b_i(x(E_j)) \y[A]_i(E_j) - w_j = b(E_j) - w_j = 0 \ ,
  \]
  which implies that
  \[
    \vec{b}\vec{A}_k = c_k \equiv 0 \pmod{u_k} \ ,
  \]
  where $\vec{A}_k \in \field[x]^{(\mu + 1) \times 1}$ denotes the $k$-th column of $\mat{A}$. But then indeed $\vec{b} \in \H_{\vec{u}}(\mat{A})$ by definition.

  Note that in the $(-\vec{d})$-degree, the leading position of $\vec{b}$ is the last position.  The $(-\vec{d})$-Popov basis of $\H_{\vec{u}}(\mat{A})$ will contain a vector $\vec{a}=(a_0,\dots,a_{\mu-1},a_\mu)$ whose leading coordinate is the last position as well, and in particular $a_{\mu} \neq 0$. Since $\vec{a}$ has minimal $(-\vec{d})$-degree among all vectors in $\H_{\vec{u}}(\mat{A})$ whose leading position is the last position, we conclude that $\vec{a}$ satisfies the same degree constraints as $\vec{b}$.

  To conclude the proof observe that
  \begin{align*}
    \delta_A(a)
    &= \max_i(\delta(a_i) + \delta_A(\y[A]_i)) \\
    &= \max_i(\mu\deg a_i + \delta_A(\y[A]_i)) \\
    &< \max_i(\mu d_i + \delta_A(\y[A]_i)) \\
    &= \deg E + 2g - \deg A \ ,
  \end{align*}
  and that for any $E_j \in U_k$, where $k = 1,\dots,\mu$, it holds that
  \begin{align*}
    a(E_j) - w_j
    &= \sum_{i=0}^{\mu-1} a_i(x(E_j))\y[A]_i(E_j) - w_j \\
    &= \sum_{i=0}^{\mu-1} a_i(x(E_j)) S_{i,k}(x(E_j)) + T_k(x(E_j)) \\
    &= (\vec{a} \vec{A}_k)(x(E_j)) = 0 \ ,
  \end{align*}
  since $\vec{a} \in \H_{\vec{u}}(\mat{A})$. Consequently, $a(E_j) = w_j$ for $j=1,\dots,N$.
\end{proof}

\begin{algorithm}[H]
  \caption{$\Interpolate(\vec{w}, E, A, \vec{x}, \vec{y})$} \label{algo:interpolate}
  \begin{algorithmic}[1]
    \Input
    \begin{itemize}
    \item Divisors $A$ and $E = E_1 + \cdots + E_N$, where $E_1,\dots,E_N \in \places_{\ffield}\setminus\{\Pinf\}$ are distinct rational places and $\supp(A) \cap \supp(E)=\emptyset$,
    \item interpolation values $\vec{w} = (w_1,\dots,w_N) \in \field^N$,
    \item evaluations $\vec{x} = (x_j)_{j=1,\dots,N}$, where $x_j = x(E_j) \in \field$,
    \item evaluations $\vec{y} = (y_{i,j})^{i=0,\dots,\mu-1}_{j=1,\dots,N}$, where $y_{i,j} = y^{(A)}_i(E_j) \in \field$.
    \end{itemize}
    \Output
    \begin{itemize}
    \item $a \in \Ya(A)$ such that $\delta_A(a) \leq \deg E + 2g - 1 - \deg A$ and $a(E_j) = w_j$ for $j=1,\dots,N$
    \end{itemize}
    \State $U_1,\dots,U_\mu \assign$ an $x$-partition of $E$
    \State $\mat{S} = [S_{i,k}] \in \field[x]^{\mu \times \mu} \assign$ matrix with $S_{i,k}(x_j) = y_{i,j}$ for all $E_j \in U_k$
    \State $\vec{T} = [T_k] \in \field[x]^\mu \assign$ row vector with $T_k(x_j) = -w_j$ for all $E_j \in U_k$
    \State $\vec{u} = (u_1,\dots,u_\mu) \in \field[x]^{\mu} \assign$ vector with $u_k = \prod_{E_j \in U_k}(x - x_j)$
    \State $\vec{d} = (d_0,\dots,d_{\mu-1},1) \in (\frac{1}{\mu}\ZZ)^{\mu+1} \assign$ vector with $d_i = \frac{1}{\mu} (\deg E + 2g - \deg A - \delta_A(\y[A]_i))$
    \State $\mat{P} \in \field[x]^{(\mu+1) \times (\mu+1)} \assign$ $(-\vec{d})$-Popov basis matrix of $\H_{\vec{u}}(\mat{A})$, where
    $
    \mat{A} =
    \left[
      \begin{array}{c}
        \mat{S} \\
        \hline
        \mat{T}
      \end{array}
    \right] \in \field^{(\mu+1) \times \mu}
    $
    \State $\vec{a} = (a_0,\dots,a_{\mu-1},1) \in \field[x]^{\mu+1} \assign$ a row of $\mat{P}$ having $1$ as its last entry and satisfying $\deg a_i < d_i$ for $i = 0,\dots,\mu-1$
    \State \Return $a = \sum_{i=0}^{\mu-1} a_i \y[A]_i$
  \end{algorithmic}
\end{algorithm}

\begin{proposition}
  \cref{algo:interpolate} is correct and costs $\softO(\mu^{\omega-1} (N + g))$ operations in $\field$.
\end{proposition}

\begin{proof}
  Correctness is given by \cref{lem:interp-correct}. For complexity observe that $\deg u_k =|U_k| \le \lceil N/\mu \rceil$ for all $k$, while for all $i,k$, we can choose $S_{i,k},T_k$ such that
  \[
    \deg S_{i,k}, \deg T_k < \lceil N/\mu \rceil \ .
  \]
  Step 2 costs $\softO(\mu^2 N/\mu) = \softO(\mu N)$. Step 3 costs $\softO(\mu N/\mu) = \softO(N)$ using fast univariate interpolation \cite[Corollary 10.12]{von_zur_gathen_modern_2012}, and Step 4 can be executed within the same cost bound using a product tree \cite[Lemma 10.4]{von_zur_gathen_modern_2012}. The computational bottleneck lies in step 6, which according to \cref{cor:sim-hermit-pade} costs
  \[
    \softO(\mu^{\omega-1} \mu (\max_i d_i + \max_k \deg u_k)) \subseteq  \softO(\mu^{\omega}(\frac{\deg E + 2g}{\mu} + \frac{N}{\mu})) = \softO(\mu^{\omega-1} (N + g)) \ .
  \]
  Here we used that $d_i \le (\deg E+2g)/\mu$, since by \cref{lem:Fx_basis}, $\deg A + \delta_A(\y[A]_i) \ge 0.$
\end{proof}

The output $a = \sum_{i=0}^{\mu-1} a_i \y[A]_i$ of $\Interpolate(\vec{w},E,A, \vec{x}, \vec{y})$ satisfies $\delta_A(a) \leq \deg E + 2g - 1 - \deg A$ as shown in \cref{lem:interp-correct}. In general this is the best one can expect, but in specific cases the existence of an interpolation function $b \in \Ya(A)$ with $\delta_A(b)<\Delta<\deg E + 2g - \deg A$ may be known to exist. The following lemma clarifies a property of the output of \cref{algo:interpolate}.

\begin{lemma}\label{lem:mindeg}
  In the context of \cref{algo:interpolate}, the output $a \in \Ya(A)$ satisfies $\delta_A(a) \le \delta_A(b)$ for all functions $b \in \Ya(A)$ with $b(E_j) = w_j$ for $j=1,\dots,N$.
\end{lemma}
\begin{proof}
  Consider the map $\varphi$ which sends any function $b = \sum_{i=0}^{\mu-1}b_iy_i^{(A)} \in \Ya(A)$ to the vector $(b_0,\dots,b_{\mu-1}) \in \field[x]^{\mu}$, and observe that if $b(E_j) = w_j$ for all $j$, then $\varphi(a-b)$ is in the row space of the matrix $\tilde{\mat{P}} \in \field[x]^{\mu \times \mu}$ obtained from the first $\mu$ rows and columns of $\mat{P}$. It is clear that $\tilde{\mat{P}}$ is in $(-\tilde{\vec{d}})$-Popov form, where $\tilde{\vec{d}} = (d_0,\dots,d_{\mu-1})$, and that each entry in $\varphi(a)$ has degree strictly smaller than the maximal degree of the corresponding column in $\tilde{\vec{P}}$: otherwise $\mat{P}$ would not be in $(-\vec{d})$-Popov form. But if each entry of $\phi(b)$ has degree no greater than the corresponding entry in $\phi(a)$, then it follows from \cref{prop:popov-properties} that $\varphi(a - b) = 0$, implying that $a=b$ since $\varphi(a - b)$ is in the row space of $\tilde{P}$ (see also \cite[Theorem 6.3-15]{kailath_linear_1980} or \cite[Lemma 1.24]{vincent_neiger_bases_2016}).
\end{proof}

\subsection{Computing a generating set over $\Ya$ of $\M$}\label{subsec:genYa}

We now return to the Guruswami-Sudan decoding of the code $\code$. In this subsection we use the symbolic expressions from \cref{cor:Ya-basis-of-M} to compute a generating set over $\Ya$ of $\M$. We start with a lemma.


\begin{lemma}\label{lem:mul-by-ev}

  Let $a \in \Ya(A)$ and $b \in \Ya(B)$, where $A$ and $B$ are divisors, and let $E = E_1 + \cdots + E_N$, where $E_1,\dots,E_N$ are distinct rational places different from $\Pinf$ and not contained in $\supp(A) \cup \supp(B)$.
  If $c \in \Ya(A + B)$ satisfies
  \begin{enumerate}
  \item $\delta_{A+B}(c) < N - \deg (A + B)$ and
  \item $c(E_j) = a(E_j)b(E_j) = (ab)(E_j)$ for $j = 1,\dots,N$,
  \end{enumerate}
  then $c = ab$.
\end{lemma}
\begin{proof}
  Note that $c \in \L(C)$, where $C = \delta_{A+B} (c)\Pinf + A + B$. The second condition simply states that $\ev_E (c) = \ev_E(ab)$, but since $\deg C < \deg E$, it follows from \cref{lem:ev} that $\ev_E: \L(C) \to \field^N$ is injective. Consequently, $c = ab$.
\end{proof}

Using \cref{algo:evaluate} and \cref{algo:interpolate}, this lemma allows us to perform efficient multiplication and hence to compute a generating set over $\Ya$ of $\M$ as in \cref{algo:basisYa}.
%
\begin{algorithm}[H]
  \caption{$\GenYa(\vec{r},D,G,E,\vec{x},\vec{y},\vec{g})$} \label{algo:basisYa}
  \begin{algorithmic}[1]
    \Input
    \begin{itemize}
    \item Received word $\vec{r} \in \field^{n}$,
    \item the code divisors $D$ and $G$,
     \item a divisor $E = E_1 + \cdots + E_N$, where $E_1,\dots,E_N$ are distinct rational places of $F$, not in $\{\Pinf\}\cup \supp G,$ such that $N \ge (\ell+1)\deg G+4g+(s+1)n$,
      \item evaluations $\vec{x} = (x_j)_{j=1,\dots,N}$, where $x_j = x(E_j) \in \field$,
      \item evaluations $\vec{y} = (y_{i,j})^{i=0,\dots,\mu-1}_{j=1,\dots,N}$, where $y_{i,j} = \y[A]_i(E_j) \in \field$,
      \item evaluations $\vec{g} = (g_{v,j}^{(u)})$, where $u=0,\dots,\ell$, $v = 1,2$ and $j = 1,\dots,N$ \newline
      such that $g_{v,j}^{(u)} = g_v^{(u)}(E_j) \in \field$ where $\langle g_1^{(u)}, g_2^{(u)} \rangle_{\Ya} = \Ya(G_u)$, and $\delta_{G_u}(g_v^{(u)}) \le 4g-1+(u+1)\deg(G)+(s+1)n$.
    \end{itemize}
    \Output
    \begin{itemize}
    \item $(B_v^{(u)})^{u=0,\dots,\ell}_{v=1,2} $ such that $\langle B_v^{(u)}\rangle_{\Ya} = \M$.
    \end{itemize}
    \State $R \in \Ya(G) \assign \Interpolate(\vec{r},D,G,\vec{x}, \vec{y})$ \Comment{\cref{algo:interpolate}}
    \State $( \hat r_1^{(0)},\dots, \hat r_N^{(0)}) \in \field^N \assign (1,\dots,1)$
    \State $( \hat r_1^{(1)},\dots, \hat r_N^{(1)}) \in \field^N \assign \Evaluate(-R,E,G,\vec{x},\vec{y})$ \Comment{\cref{algo:evaluate}}
    \For{$u = 2,\dots,\ell$}
    \State $( \hat r_1^{(u)},\dots, \hat r_N^{(u)}) \in \field^N \assign
    ( \hat r_1^{(1)} \hat r_1^{(u-1)},\dots, \hat r_N^{(1)} \hat r_N^{(u-1)})$
    \EndFor
    \For{$u=0,\dots,\ell$, $r = 0,\dots,u$ and $v = 1,2$}
    \State $\vec{c}_{r,v}^{(u)} \in \field^N \assign
    (\hat r_1^{(u-r)}g_{v,1}^{(u)}, \dots, \hat r_N^{(u-r)}g_{v,N}^{(u)})$
    \State $c_{r,v}^{(u)} \in \Ya(-rG) \assign \Interpolate(\vec{c}_{r,v}^{(u)},E,-rG, \vec{x}, \vec{y})$
    \EndFor
    \For{$u=0,\dots,\ell$ and $v = 1,2$}
    \State $B_v^{(u)} \in \M \assign \sum_{r=0}^u\binom{u}{r}z^rc_{r,v}^{(u)}$
    \EndFor
    \State \Return $(B_v^{(u)})^{u=0,\dots,\ell}_{v=1,2}$
  \end{algorithmic}
\end{algorithm}
\begin{proposition}\label{prop:basisYaalg}
  \cref{algo:basisYa} is correct and costs $\softO(\ell^3\mu^{\omega-1}(n+g))$.
\end{proposition}
\begin{proof}
For correctness first observe that the postulated $g_v^{(u)}$ exist by \cref{cor:Ya-basis-of-M}. 

Note that $\delta_G(R) \le n+2g-1-\deg G$. Using the given upper bound for $\delta_{G_u}(g_v^{(u)})$, we obtain that
\begin{align}\label{eq:boundsizecrv}
\delta_{-rG}(R^{u-r}g_v^{(u)}) &= \delta_{(u-r)G+G_u}(R^{u-r}g_v^{(u)}) \notag \\
& \le (u-r)(n+2g-1-\deg G)+4g-1+(u+1)\deg G+(s+1)n \notag \\
& = (r+1)\deg G+(u-r+2)(2g-1)+1+(s+1)n \\
& = (r+1)(\deg G-2g+1)+(u+3)(2g-1)+1+(s+1)n \notag \\
& \le (\ell+1)(\deg G-2g+1)+(\ell+3)(2g-1)+1+(s+1)n \notag \\
& = (\ell+1)\deg G+2(2g-1)+1+(s+1)n \notag \\
& < (\ell+1)\deg G+4g+(s+1)n \notag 
\end{align}
\cref{lem:mindeg} then implies that $\Interpolate(\vec{c}_{r,v}^{(u)},E,-rG, \vec{x}, \vec{y})$ will output a function $c_{r,v}^{(u)} \in \Ya(G)$ satisfying $\delta_{-rG}(c_{r,v}^{(u)}) < (\ell+1)\deg G+4g+(s+1)n.$

To complete the correctness proof, we consider \cref{lem:mul-by-ev} for the divisors $A=(u-r)G$, $B=-uG$ and the function $a=(-R)^{u-r}$, $b=g_v^{(u)}$, and $c=c_{r,v}^{(u)}$. By construction, it is clear that for all $E_j \in \supp E$ we have $c_{r,v}^{(u)}(E_j)=(-R)^{u-r}(E_j)g_v^{(u)}(E_j).$
Moreover, $\deg E \ge (\ell+1)\deg G+4g+(s+1)n$, whence $\delta_{-rG}(c_{r,v}^{(u)}) < \deg E \le \deg E - \deg(-rG)$. Hence \cref{lem:mul-by-ev} implies $c_{r,v}^{(u)}=(-R)^{u-r}g_v^{(u)}.$

The complexity of the algorithm is dominated by the for loop in Lines 6--8. The $\bigO(\ell^2)$ calls of the algorithm $\Interpolate(\vec{c}_{r,v}^{(u)},E,-rG, \vec{x}, \vec{y})$ cost $\ell^2\softO(\ell\mu^{\omega-1}(n+g))$ operations.
Hence the total complexity is $\softO(\ell^3\mu^{\omega-1}(n+g)).$
\end{proof}

\begin{remark}
The generating set consisting of $\tilde{B}_v^{(u)}$ as described in \cref{rem:alternative-Fqx-basis-of-M}, can be computed slightly faster. Indeed, since in these generators, the needed powers $(-R)^u$ have the range $u=0,\dots,s$, the for loop in Lines 6--9 has  $\bigO(s\ell)$ calls of the algorithm $\Interpolate(\vec{c}_{r,v}^{(u)},E,-rG, \vec{x}, \vec{y})$. Hence to compute the $\tilde{B}_v^{(u)}$ costs $\softO(s\ell^2\mu^{\omega-1}(n+g)).$
\end{remark}

\subsection{Computing a generating set over $\field[x]$ of $\M$.}\label{subsec:genFqx}

In the previous subsection, we saw how to efficiently compute the generating set $\{B_v^{(u)}\}$ of $\M$ over $\Ya$, as in \cref{cor:Ya-basis-of-M}. The next logical step is to compute the set of products $\{y_iB_v^{(u)}\}$, which generates $\M$ over $\field[x]$ according to \cref{cor:Fqx-basis-of-M}. Consequently, we now consider the following problem: given a function $a \in \Ya(A)$ for some divisor $A$, compute $y_0a,\dots,y_{\mu-1}a \in \Ya(A)$. Computing the $y_ia$ individually using \cref{algo:basisYa} would be too slow for our purposes. Indeed, obtaining each $y_iB_v^{(u)}$ this way would cost $\softO(\ell^2\mu^{\omega-1}(n+g))$ operations, and we need to compute $2\mu(\ell+1)$ such terms in total. Therefore, we introduce in this subsection a more efficient approach, which will allow us to compute $y_0a,\dots,y_{\mu-1}a$ simultaneously.

\begin{definition}
  For any $H(z) \in \ffield[z]$ and any rational place $P \in \places_{\ffield}$ that is not a pole of any of the coefficients of $H(z)$, and $\alpha \in \field$ we denote by $H(P,\alpha)$ the evaluation of $H(\alpha) \in \ffield$ at $P$.
\end{definition}

\begin{definition}\label{def:Nae}
Let $A$ be a divisor and $E = E_1 + \cdots + E_N$ for distinct rational places $E_1,\dots,E_N$ of $F$ different from $\Pinf$ such that $\supp(A) \cap \supp(E)= \emptyset$. For $a\in \Ya(A)$, we define the $\field[x]$-module
  \[
    \N_{A,E}(a) = \{
    H = H_0 + H_1 z \in \Ya(A) \oplus z\Ya \mid H(P,a(P)) = 0 \text { for all } P \in \supp(E)
    \} \ .
  \]
\end{definition}

In the following lemmas, we use the same notation $A$, $E$ as in \cref{def:Nae}.


\begin{lemma}
  \label{lem:fxbasis_basismem}
  Let $a \in \Ya(A)$. If $H = H_0 + z H_1 \in \N_{A,E}(a)$ with
  \[
    \max\{
    \delta_A(H_0),
    \delta(H_1) + \delta_A(a)
    \} < \deg E - \deg A \ ,
  \]
then $H(a) = 0$, i.e. $H \in \langle z-a \rangle_{\Ya}$.
\end{lemma}
\begin{proof}
Since $H \in \N_{A,E}(a)$, we have $H(a) \in \Ya(A)$. Hence by definition of $\delta_A$, we have $H(a) \in \L(A+\delta_A(H(a))\Pinf).$
Since for all $E_j \in \supp E$, we have $H(a)(E_j)=0$ and $\supp E \cap (\supp A \cup \{\Pinf\})= \emptyset$, we may conclude that
$H(a) \in \L(A+\delta_A(H(a))\Pinf-E).$
Moreover,
\begin{align*}
    \delta_A(H(a))
    \leq \max\{
    \delta_A(H_0),
    \delta(H_1) + \delta_A(a)
    \}
    < \deg E - \deg A \ ,
\end{align*}
which ensures that the aforementioned Riemann-Roch space is trivial.
\end{proof}

\begin{lemma}
  \label{lem:fxbasis_pade}
  Let $a \in \Ya(A)$. Furthermore, let $U_1,\dots,U_\mu$ be an $x$-partition of $E$, and let $\mat{S} = [S_{i,k}], \mat{T} = [T_{i,k}]$ be matrices in $\field[x]^{\mu \times \mu}$ such that
  \begin{align*}
   S_{i,k}(x(E_j)) = \y[A]_i(E_j) \and
   T_{i,k}(x(E_j)) = a(E_j)\y_i(E_j) \for E_j \in U_k \ .
  \end{align*}
  If $\vec{u} = (u_1,\dots,u_\mu) \in \field[x]^\mu$, where $u_k = \prod_{E_j \in \supp U_k}(x - x(E_j))$, then the map
  \[
    \psi: \sum_{i = 0}^{\mu-1}(s_i \y[A]_i + t_i z \y_i) \mapsto (s_0,\dots,s_{\mu-1},t_0,\dots,t_{\mu-1})
  \]
  is an $\field[x]$-isomorphism between $\N_{A,E}(a)$ and $\H_{\vec{u}}(\mat{A})$, where
  \[
    \mat{A} =
    \left[
      \begin{array}{c}
        \mat{S} \\
        \hline
        \mat{T}
      \end{array}
    \right] \in \field^{2\mu \times \mu} \ .
  \]
\end{lemma}
\begin{proof}
  Clearly $\psi$ is an $\field[x]$-isomorphism between $\Ya(A) \oplus z \Ya$ and $\field[x]^{2\mu}$, therefore it suffices to show that for any $H \in \Ya(A) \oplus z \Ya$ it holds that $H \in \N_{A,E}(a)$ if and only if $\psi(H) \in \H_{\vec{q}}(\mat{A})$, i.e. that $H(E_j,a(E_j)) = 0$ for all $E_j \in \supp U_k$ and all $k=1,\dots,\mu$ if and only if $\psi(H) \cdot \mat{A}_k \equiv 0 \mod u_k$, for $k = 1,\dots,\mu$, where $\vec{A}_k$ denotes the $k$-th column of $\mat{A}$. But this is necessarily true, since for every $E_j \in U_k$ the following identity holds, where $\alpha = x(E_j)$:
  \begin{align*}
    H(E_j, a(E_j))
    &= \sum_{i=0}^{\mu-1}\big( s_i(\alpha) \y[A]_i(E_j) + a(E_j)t_i(\alpha)\y_i(E_j) \big ) \\
    &= \sum_{i=0}^{\mu-1}\big( s_i (\alpha) S_{i,k}(\alpha) + t_i(\alpha)T_{i,k}(\alpha) \big)
    = (\psi(H) \cdot \mat{A}_k)(\alpha) \ .
  \end{align*}
\end{proof}

\begin{lemma}
  \label{lem:popov-submatrix}
  In the context of \cref{lem:fxbasis_pade}, if $\mat{P} \in \field[x]^{2\mu \times 2\mu}$ is the $\vec{d}$-Popov basis of $\H_{\vec{u}}(\mat{A}) = \psi(\N_{A,E}(a))$, where $\deg E \geq 2g  + \mu + \delta_A(a) + \deg A$ and
  \[
    \vec{d} = \frac{1}{\mu}(\delta_A(\y[A]_0),\dots,\delta_A(\y[A]_{\mu-1}),
    \delta(\y_0) + \delta_A(a),\dots,\delta(\y_{\mu-1}) + \delta_A(a)) \in (\frac{1}{\mu}\ZZ)^{2\mu} \ ,
  \]
  then exactly $\mu$ rows of $\mat{P}$ have $\vec{d}$-degree less than $\frac{1}{\mu}(\deg E - \deg A)$.
  Furthermore, if $\tilde{\mat{P}} \in \field[x]^{\mu \times 2 \mu}$ is the submatrix of $\mat{P}$ consisting of these rows, then the $k$-th row of $\tilde{\mat{P}}$ is $\psi(Y_k)$ for $k = 1,\dots,\mu$, where $Y_k = -a\y_{k-1} +z\y_{k-1} \in \langle z - a \rangle_{\Ya} \subset \N_{A,E}(a)$. Consequently, if $\tilde{\mat{P}} = [\tilde{\mat{P}}_1|\tilde{\mat{P}}_2]$, where $\tilde{\mat{P}}_1,\tilde{\mat{P}}_2 \in \field[x]^{\mu \times \mu}$, then $a\y_{k-1} = -\sum_{i=0}^{\mu-1} p_{k,i} \y[A]_i$, where $(p_{k,0},\dots,p_{k,\mu-1})$ is the $k$-th row of $\tilde{\mat{P}}_1$.
\end{lemma}
\begin{proof}
  For any
  \[
    H = H_0 + z H_1 \in \N_{A,E}(a) \ ,
  \]
  where $H_0 = \sum_{i=0}^{\mu-1} s_i\y[A]_i \in \Ya(A)$ and $H_1 = \sum_{i=0}^{\mu-1} t_i\y_i \in \Ya$ with $s_i,t_i \in \field[x]$, it holds that
  \begin{align*}
    \deg_{\vec{d}} \psi(H)
    &= \max\{
      \max_i ( \deg s_i + \frac{\delta_A(\y[A]_i)}{\mu} ),
      \max_i( \deg t_i + \frac{\delta(\y_i) + \delta_A(a)}{\mu} )
      \} \\
    &= \frac{1}{\mu}\max\{
      \delta_A(H_0),
      \delta(H_1) + \delta_A(a)
      \} \ .
  \end{align*}
  It then follows from \cref{lem:fxbasis_basismem} that
  \[
    \deg_{\vec{d}} \psi(H) < \frac{1}{\mu}(\deg E - \deg A) \implies H \in \langle z-a \rangle_{\Ya} \ ,
  \]
  which means that at most $\mu$ rows of $\mat{P}$ can have $\vec{d}$-degree less than $\frac{1}{\mu}(\deg E - \deg A)$, because $\langle z-a \rangle_{\Ya}$ has rank $\mu$ as an $\field[x]$-module. On the other hand, since $Y_1,\dots,Y_\mu$ are linearly independent over $\field[x]$, and since
  \[
    \deg_{\vec{d}} \psi(Y_k) = \frac{1}{\mu} (\delta(\y_{k-1}) + \delta_A(a)) < \frac{1}{\mu}(\delta_A(a) + 2g  + \mu) \leq \frac{1}{\mu}(\deg E - \deg A)
  \]
  for $k = 1,\dots,\mu$, where the strict inequality is due to \cref{lem:Fx_basis}, then at least $\mu$ rows of $\mat{P}$ have $\vec{d}$-degree less than $\frac{1}{\mu}(\deg E - \deg A)$, because $\mat{P}$ is $\vec{d}$-row reduced. This proves the first claim of the lemma.

 For the second claim it is sufficient to show that the $\vec{d}$-pivot index of $\psi(Y_k)$ is $\mu + k$, since this would imply that the matrix whose rows are $\psi(Y_k)$ is in $\vec{d}$-Popov form. To see this, write $Y_k = -\sum_{i=0}^{\mu-1} w_i\y[A]_i + z \y_{k-1}$, where $w_i \in \field[x]$, and note that $Y_k(a) = 0$ implies that
 \[
    \max_i \delta_A(w_i\y[A]_i) = \delta_A(\sum_{i=0}^{\mu-1} w_i\y[A]_i) =  \delta_A(a \y_{k-1}) = \delta(\y_{k-1}) + \delta_A(a) \ .
  \]
  Consequently, $\deg_{\vec{d}} \psi(Y_k) = \frac{1}{\mu}(\delta(\y_{k-1}) + \delta_A(a))$, which shows that $\mu + k$ is indeed the $\vec{d}$-pivot index of $\psi(Y_k)$.
\end{proof}

\begin{algorithm}[H]
  \caption{$\BasisProducts(a,E,A,\vec{x},\vec{y})$} \label{algo:basis-products}
  \begin{algorithmic}[1]
    \Input
    \begin{itemize}
    \item A divisor $A$,
    \item a function $a \in \Ya(A)$,
    \item a divisor $E = E_1 + \cdots + E_N$, where $E_1,\dots,E_N \in \places_{\ffield}\setminus\{\Pinf\}$ are distinct rational places, $\supp(A) \cap \supp(E)=\emptyset$ and $\deg E \geq \deg A + \delta_A(a) + 2g  + \mu$,
    \item evaluations $\vec{x} = (x_j)_{j=1,\dots,N}$, where $x_j = x(E_j) \in \field$,
    \item evaluations $\vec{y} = (y_{i,j})^{i=0,\dots,\mu-1}_{j=1,\dots,N}$, where $y_{i,j} = y^{(A)}_i(E_j) \in \field$.
    \end{itemize}
    \Output
    \begin{itemize}
    \item Products $(ay_0,\cdots,ay_{\mu-1})$, where each $ay_i \in \Ya(A)$.
    \end{itemize}
    \If{$a=0$}
    \State \Return $(0,\dots,0)$
    \EndIf
    \State $U_1,\dots,U_\mu \assign $ an $x$-partition of $E$
    \State $\mat{S} = [S_{i,k}] \in \field[x]^{\mu \times \mu} \assign$
    matrix with $S_{i,k}(x_j) = y_{i,j} \text{ for } E_j \in U_k$
    \State $\mat{T} = [T_{i,k}] \in \field[x]^{\mu \times \mu} \assign$
    matrix with $T_{i,k}(x_j) = a(E_j)y_{i,j} \text{ for } E_j \in U_k$
    \State $\vec{u} = (u_1,\dots,u_\mu) \in \field[x]^{\mu} \assign$
    vector with $u_k = \prod_{E_j \in U_k}(x - x_j)$
    \State $\vec{d} \in (\frac{1}{\mu}\ZZ)^{2\mu} \assign
    \frac{1}{\mu} \big(
    \delta_A(y^{(A)}_0),\dots,\delta_A(y^{(A)}_{\mu-1}),
    \delta(y_0) + \delta_A(a),\dots,\delta(y_{\mu-1}) + \delta_A(a)
    \big)$
    \State $\mat{P} \in \field[x]^{2\mu \times 2\mu} \assign$ $\vec{d}$-Popov basis of $\H_{\vec{u}}(\mat{A})$, where
    $\mat{A} =
    \left[
      \begin{array}{c}
        \mat{S} \\
        \hline
        \mat{T}
      \end{array}
    \right] \in \field[x]^{2\mu \times \mu}$
    \State $[ \tilde{\mat{P}}_1 |  \tilde{\mat{P}}_2] \in \field[x]^{\mu \times 2\mu} \assign$ the submatrix of $\mat{P}$ consisting of all rows with
    \Statex \indent $\vec{d}$-degree less than $\frac{1}{\mu}(\deg E - \deg A)$, where $\tilde{\mat{P}}_1, \tilde{\mat{P}}_2 \in \field[x]^{\mu \times \mu}$
    \For{$k=1,\dots,\mu$}
    \State $(p_{k,0},\dots,p_{k,\mu-1}) \in \field[x]^\mu \assign$ $k$-th row of $\mat{P}_1$
    \State $a_k \in \Ya(A) \assign -\sum_{i=0}^{\mu-1} p_{k,i} y^{(A)}_i$
    \EndFor
    \State \Return $(a_1,\dots,a_\mu)$
  \end{algorithmic}
\end{algorithm}

\begin{lemma}\label{lem:basis-products}
  \cref{algo:basis-products} is correct and costs $\softO(\mu^{\omega-1}(N+|\deg A|))$ operations in $\field$.
\end{lemma}
\begin{proof}
Correctness is given by \cref{lem:popov-submatrix}. For complexity, simply note that the computational bottleneck lies in Step 8, in which case $\delta_A(a) \ge -\deg A$ because $a$ is nonzero and $a \in \L(\delta_A(a)\Pinf + A)$. By assumption, we have that $N=\deg E \geq \deg A + \delta_A(a) + 2g  + \mu$,
hence by \cref{lem:Fx_basis}
\begin{align*}
  -\deg A &\le \delta_A(y_i^{(A)})
  \le 2g-1-\deg A +\mu \\
  &< \deg E- 2\deg A -\delta_A(a) \le \deg E-\deg A \ .
\end{align*}
Since $\deg u_k \le N/\mu$ for $k=1,\dots,\mu$, then the total complexity of the algorithm is given by
\cref{cor:sim-hermit-pade} as
\[
  \softO \big(
  \mu^{\omega-1}\max\{|\deg E|,|\deg E-\deg A|,|\deg A|\}
  \big) \subseteq \softO(\mu^{\omega-1}(N+|\deg A|))
\]
operations in $\field$.
\end{proof}

Now we are ready to state \cref{algo:basisFqx}, which computes a generating set over $\field[x]$ of $\M$.

\begin{algorithm}[H]
  \caption{$\GenFqx(\vec{r},D,G,E,\vec{x},\vec{y},\vec{g})$} \label{algo:basisFqx}
  \begin{algorithmic}[1]
    \Input
    \begin{itemize}
    \item Received word $\vec{r} \in \field^{n}$,
    \item divisors $D$ and $G$ for the code $\code$,
    \item a divisor $E = E_1 + \cdots + E_N$, where $E_1,\dots,E_N \in \places_{\ffield}\setminus\{\Pinf\}$ are distinct rational places, $\supp(A) \cap \supp(E)=\emptyset$ and $N \ge \max\{\deg G+(\ell+3)(2g-1)+(s+1)n+2+\mu,
      \newline \hphantom{ N \ge \max\{ } (\ell+1)\deg G+4g+(s+1)n\}$,
    \item evaluations $\vec{x} = (x_j)_{j=1,\dots,N}$, where $x_j = x(E_j) \in \field$,
    \item evaluations $\vec{y} = (y_{i,j})^{i=0,\dots,\mu-1}_{j=1,\dots,N}$, where $y_{i,j} = y^{(A)}_i(E_j) \in \field$,
      \item evaluations $\vec{g} = (g_{v,j}^{(u)})^{u=0,\dots,\ell}_{v=1,2, \ j=1,\dots,N}$,
      where $g_{v,j}^{(u)} = g_v^{(u)}(E_j) \in \field$, $\langle g_1^{(u)}, g_2^{(u)} \rangle_{\Ya} = \Ya(G_u)$ and $\delta_{G_u}(g_v^{(u)}) \le 4g-1+(u+1)\deg(G)+(s+1)n$,
      as in \cref{cor:Ya-basis-of-M}.
    \end{itemize}
    \Output
    \begin{itemize}
    \item $(y_iB_v^{(u)})^{u=0,\dots,\ell}_{i=0,\dots,\mu-1, \ v =1,2}$, where the $B_v^{(u)} \in \M$ are as in
      \newline\cref{cor:Ya-basis-of-M}, i.e. $\langle y_iB_v^{(u)} \rangle_{\field[x]} = \M$.
    \end{itemize}
    \State $(B_v^{(u)})^{u=0,\dots,\ell}_{v=1,2} \assign \GenYa(\vec{r},D,G,E,\vec{x},\vec{y},\vec{g})$ \Comment{\cref{algo:basisYa}}
    \For{$u=0,\dots,\ell$, $v = 1,2$ and $t = 0,\dots,u$}
    \State $b_{v,t}^{(u)} \in \Ya(-tG) \assign$ the $z^t$-coefficient of $B_v^{(u)}$
    \State $(y_ib_{v,t}^{(u)})_{i=0,\dots,\mu-1} \assign \BasisProducts(b_{v,t}^{(u)},E,-tG,\vec{x},\vec{y})$ \Comment{\cref{algo:basis-products}}
    \EndFor
    \For{$u=0,\dots,\ell$, $v = 1,2$ and $i=0,\dots,\mu-1$}
    \State $B_{v,i}^{(u)} \in \M \assign \sum_{t=0}^uz^t y_i b_{v,t}^{(u)}$
    \EndFor
    \State \Return $(B_{v,i}^{(u)})^{u=0,\dots,\ell}_{v=1,2, \ i=0,\dots,\mu-1}$
  \end{algorithmic}
\end{algorithm}

\begin{proposition}
  \cref{algo:basisFqx} is correct and costs $\softO(\ell^3\mu^{\omega-1}(n+g))$ operations in $\field$.
\end{proposition}
\begin{proof}
Correctness follows immediately from \cref{cor:Fqx-basis-of-M} and \cref{lem:basis-products} once we show that the calls $\BasisProducts(b_{v,t}^{(u)},E,-tG,\vec{x},\vec{y})$ in Line 4 are valid. In particular, we need to verify that
\begin{equation}
  \label{eq:N-bound}
  N \ge \deg(-tG)+\delta_{-tG}(b_{v,t}^{(u)})+2g+\mu
\end{equation}
for all appropriate values of $u,v$ and $t$. Using the notation from \cref{cor:Ya-basis-of-M} and \cref{algo:basisYa}, we know that $b_{v,t}^{(u)} = \binom{u}{t}(-R)^{u-t}g_v^{(u)}$, hence by \eqref{eq:boundsizecrv}
\begin{equation}\label{eq:sizebrvu}
\delta_{-tG}(b_{r,v}^{(u)}) \le (t+1)\deg G+(u-t+2)(2g-1)+(s+1)n+1 \ .
\end{equation}
The sought bound \eqref{eq:N-bound} on $N$ then follows from
\[
  -t\deg G +\delta_{-tG}(b_{v,t}^{(u)}) \le \deg G +(\ell+2)(2g-1)+(s+1)n+1 \ .
\]
For the complexity, we note that Line 1 costs $\softO(\ell^3\mu^{\omega-1}(n+g))$ operations by \cref{prop:basisYaalg}, while each call $\BasisProducts(b_{v,t}^{(u)},E,-tG,\vec{x},\vec{y})$ in Line 4 costs $\softO(\mu^{\omega-1}(N+|\deg(-tG)|)) \subseteq \softO(\ell\mu^{\omega-1}(n+g))$ operations by \cref{lem:basis-products}. Since the for-loop in Line 2 has $\bigO(\ell^2)$ iterations, the stated complexity follows\,--\,the rest of the algorithm is memory management and is therefore ``free''.
\end{proof}

\begin{remark}
Computing the generating set $\{y_i\tilde{B}_v^{(u)}\}$ over $\field[x]$ of $\M$ can be done in $\softO(s\ell^2\mu^{\omega-1}(n+g))$, since in that case only $\bigO(s\ell)$ coefficients of the $\tilde{B}_v^{(u)}$ are nonzero.
\end{remark}

\subsection{Finding a nonzero $Q \in \M$ satisfying $\delta_G(Q)<s(n-\tau)$}\label{subsec:findQ}

The following lemma introduces notation that may be needed to describe the decoding algorithm.

\begin{lemma}
  \label{lem:Fx-mat}
  For any divisor $A$ and any $a = \sum_{i=0}^{\mu-1} a_i \y[A]_i \in \Ya(A)$, where $a_i \in \field[x]$, let
  \[
    \curlyvee^{(A)}(a) = (a_0,\dots,a_{\mu-1}) \in \field[x]^{\mu} \ ,
  \]
  and for any $Q = \sum_{t=0}^\ell z^t Q^{(t)} \in \bigoplus_{t=0}^{\ell}z^t\Ya(-tG)$ let
  \[
    \curlyvee_z(Q) = (\curlyvee^{(0)}(Q^{(0)})|\curlyvee^{(-G)}(Q^{(1)})|\cdots|\curlyvee^{(-\ell G)}(Q^{(\ell)})) \in \field[x]^{\mu(\ell+1)} \ .
  \]

  If $B_v^{(u)} \in \bigoplus_{t=0}^{\ell}z^t\Ya(-tG)$ for $u=0,\dots,\ell$ and $v = 1,2$ are as in \cref{cor:Ya-basis-of-M} and
  \[
    \mat{M}_{s,\ell} =
    \left[
    \begin{array}{c}
      \mat{M}_{s,\ell}^{(1)} \\
      \hline
      \mat{M}_{s,\ell}^{(2)}
    \end{array}
    \right] \in \field[x]^{2\mu(\ell+1) \times \mu(\ell+1)} \ ,
  \]
  where for $v=1,2$
  \[
    \mat{M}_{s,\ell}^{(v)} =
    \begin{pmatrix}
      \left[
        \begin{array}{c}
          \curlyvee_z(y_0 B_v^{(0)}) \\
          \hline
          \vdots \\
          \hline
          \curlyvee_z(y_{\mu-1} B_v^{(0)})
        \end{array}
      \right]^\top
      \ \vline \ &\cdots& \vline \ \;
      \left[
        \begin{array}{c}
          \curlyvee_z(y_0 B_v^{(\ell)}) \\
          \hline
          \vdots \\
          \hline
          \curlyvee_z(y_{\mu-1} B_v^{(\ell)})
        \end{array}
      \right]^\top
\end{pmatrix}^\top \ ,
  \]
  then $\curlyvee_z$ is an $\field[x]$-isomorphism between $\M$ and the row space of ${\mat M}_{s,\ell}$. Moreover, for any $Q$ as before, it holds that $\delta_G (Q) = \mu \deg_{\vec{d}} \curlyvee_z(Q)$, where
  $\vec{d} = (\vec{d}^{(0)}|\cdots|\vec{d}^{(\ell)}) \in (\frac{1}{\mu}\ZZ)^{\mu(\ell + 1)}$ with
  $
  \vec{d}^{(t)} = \frac{1}{\mu}\big(
    \delta_{-tG}(\y[-tG]_0), \dots, \delta_{-tG}(\y[-tG]_{\mu-1})
    \big) \in (\frac{1}{\mu}\ZZ)^\mu
  $
  for $t = 0,\dots,\ell$.
\end{lemma}
\begin{proof}
\cref{cor:Fqx-basis-of-M} immediately implies that $\curlyvee_z$ is an $\field[x]$-isomorphism between $\M$ and the row space of ${\mat M}_{s,\ell}$. Further, writing $Q^{(t)} = \sum_{i=0}^{\mu-1} Q_i^{(t)}\y[-tG]_i$ for $t = 0,\dots,\ell$, where $Q_i^{(t)} \in \field[x]$, gives that
  \begin{align*}
    \delta_G(Q) &= \max_t \delta_{-tG}(Q^{(t)}) \\
    &= \max_{t,i} \{ \delta_{-tG}(Q_i^{(t)}\y[-tG]_i) \} \\
    &= \max_{t,i} \{ \delta (Q_i^{(t)}) + \delta_{-tG}(\y[-tG]_i) \} \\
    &= \max_{t,i} \{ \mu \deg Q_i^{(t)} + \delta_{-tG}(\y[-tG]_i)\} \\
    &= \mu \deg_{\vec{d}} \curlyvee_z (Q) \ .
  \end{align*}
\end{proof}

\cref{lem:Fx-mat} implies that we can find a nonzero $Q \in \M$ satisfying $\delta_G(Q)<s(n-\tau)$, if it exists, by computing the $\vec{d}$-Popov form of the matrix $\mat{M}_{s,\ell} \in \field[x]^{2\mu(\ell + 1) \times \mu(\ell + 1)}$. According to \cref{cor:rect-popov-cost}, this can be achieved with cost $\softO(\ell^\omega \mu^\omega \deg \mat{M}_{s,\ell})$. To estimate $\deg \mat{M}_{s,\ell}$, observe that \cref{lem:deg} implies that
\[
  \deg {\mat M}_{s,\ell} \le \frac{1}{\mu}\max_{i,r,v,u}\{-r\deg G+\delta(y_i)+\delta_{-rG}(b_{r,v}^{(u)})\} \ .
\]
Then \cref{lem:Fx_basis} and inequality \eqref{eq:sizebrvu} imply that
\begin{equation}\label{eq:maxdegMsl}
\deg {\mat M}_{s,\ell} \le \max_{r,u}\frac{6g-2+\mu+(u-r)(n+2g-1)+(s+1)n+\deg G}{\mu} \in \bigO(\mu^{-1}\ell(n+g)) \ ,
\end{equation}
which means that we can compute the $\vec{d}$-Popov form of $\mat{M}_{s,\ell}$ within our target complexity $\softO(l^{\omega + 1} \mu^{\omega - 1}(n+g))$.

\begin{remark}
Using the alternative generating set from \cref{rem:alternative-Fqx-basis-of-M}, we again get an improvement on the running time. In equation \eqref{eq:maxdegMsl}, the expression $u-r$ corresponded to the exponent of $-R$ in the expression $\binom{u}{r}(-R)^{u-r}g_v^{(u)}$, which was the coefficient of $z^r$ in $B_v^{(u)}$. Since the exponent of $-R$ in a coefficient of $\tilde{B}_v^{(u)}$ never exceeds $s$, we therefore obtain from equation \eqref{eq:maxdegMsl} the improved complexity $\softO(s\ell^\omega \mu^{\omega-1}(n+g)).$
\end{remark}

\subsection{Root-finding}\label{subsec:findroots}

In this subsection, we consider the final computational ingredient that we will need for Guruswami-Sudan list-decoding: given a polynomial $Q(z) \in \M$, compute the set $L = \{f \in \L(G) \mid Q(f) = 0 \}$ of all roots of $Q$. We accomplish this by changing the representation of $Q$ from $\bigoplus_{t=0}^\ell z^t \Ya(-tG)$ to $\bigoplus_{t=0}^\ell z^t \field[\![x]\!]$, which will allow us to use the root-finding algorithm from \cite{neiger_fast_2017}.

Let $P_0 \not \in \supp G \cup \{\Pinf\}$ be the fixed rational place of $\ffield$ for which $x$ is a local parameter. For any nonzero $h \in \ffield$ let $\pow{h} \in x^{\val[P_0](h)} \field[\![x]\!]$ denote the $P_0$-adic power series expansion of $h$ in $x$ and define $\pow{0}=0$. Furthermore, for any $Q = \sum_{t} z^t Q^{(t)} \in \ffield[z]$ let $\pow{Q} = \sum_{t} z^t \pow{Q}^{(t)}$. Recall that if $Q^{(t)} \in \Ya(-tG)$ for all $t$, then $\delta_{G}(Q) = \max_{t} \delta_{-tG}Q^{(t)}$. The following definition is from \cite{neiger_fast_2017}, and it describes the output of their root-finding algorithm:

\begin{definition}
  \label{def:basic-root-set}
  If $\pow{Q} \in \field[\![x]\!][z]$ and $\beta \in \ZZ_{\ge 0}$, then a \emph{basic root set} of $\pow{Q}$ to precision $\beta$ is a set $\{(\pow{f}_r,\alpha_r)\}_{r=1}^m \subset \field[x]\times\ZZ_{\ge 0}$ with $m \le \deg \pow{Q}$ such that
  \begin{enumerate}[1)]
  \item $\pow{Q}(\pow{f}_r + x^{\alpha_r}z) \equiv 0 \pmod{x^\beta}$ for $r=1,\dots,m$, and
  \item $\pow{Q}(\pow{f}) \equiv 0 \pmod {x^\beta} \iff \pow{f} \in \bigcup_{r=1}^m(\pow{f}_r + x^{\alpha_r}\field[\![x]\!])$ for every $\pow{f} \in \field[\![x]\!]$.
  \end{enumerate}
\end{definition}

Our algorithm for computing the sought roots of $Q \in \M$ will fundamentally rely on the following result:

\begin{theorem}[{\cite[Theorem 1.2]{neiger_fast_2017}}]
  \label{thm:neiger-root-find}
  There is an algorithm which for any $\pow{Q} \in \field[\![x]\!][z]$ and any precision $\beta \in \ZZ_{\ge 0}$ computes a basic root set of $\pow{Q}$ to precision $\beta$ using $\softO(\ell \beta)$ deterministic operations in $\field$, together with an extra $\softO(\algo{R}_{\field}(\ell) \beta)$ operations, where $\algo{R}_{\field}(\ell)$ is the cost of finding all $\field$-roots of a degree $\ell$ polynomial in $\field[z]$. Here, we can choose to use a Las Vegas algorithm with $\algo{R}_{\field}(\ell) \in \softO(\ell)$, e.g. \cite[Corollary 14.16]{von_zur_gathen_modern_2012}, or a deterministic one from \cite{shoup1991fast} with $\algo{R}_{\field}(\ell) \in \softO(\ell \kappa^2 \sqrt{p})$, where $|\field| = p^\kappa$ for some prime $p$.
\end{theorem}
In order to use \cref{thm:neiger-root-find} in our setting, we will need to address the following:
\begin{enumerate}[1)]
\item how to choose the precision $\beta$,
\item how to convert $Q \in \bigoplus_{t=0}^\ell z^t \Ya(-tG)$ to $\pow{Q} \in \bigoplus_{t=0}^\ell z^t \field[\![x]\!]$ and
\item how to obtain the roots $f \in \L(G)$ of $Q$ from a basic root set of $\pow{Q}$.
\end{enumerate}

The second item in the above list is the simplest\,--\,writing $Q = \sum_{t=0}^\ell z^t Q^{(t)}$ with $Q^{(t)} = \sum_{i=0}^{\mu-1}Q^{(t)}_iy^{(-tG)}_i$, where $Q^{(t)}_i \in \field[x]$, we can compute $\pow{Q} = \sum_{t=0}^\ell z^t \pow{Q}^{(t)}$ by simply relying on the identity $\pow{Q}^{(t)} = \sum_{i=0}^{\mu-1}Q^{(t)}_i\pow{y}^{(-tG)}_i$. Assuming that we have precomputed the $\pow{y}^{(-tG)}_i \in \field[\![x]\!]$ to sufficiently high precision, this is just basic arithmetic in $\field[x]$.

When it comes to the choice of the precision $\beta$, then there are two restrictions that ought to be considered. The first one comes from making sure that we don't return ``spurious'' roots, i.e. those $f \in \L(G)$ such that $\pow{Q}(\pow{f}) \equiv 0 \pmod{x^{\beta}}$ while $Q(f) \ne 0$. As we are about to see in the following lemma, this issue is easily avoided by choosing $\beta > \delta_G (Q)$.

\begin{lemma}\label{lem:aprox-root-is-root}
  Let $Q(z) = \sum_{t=0}^{\ell} z^t Q^{(t)} $ with $Q^{(t)} \in \Ya(-tG)$, and let $f \in \L(G)$. If $\beta > \delta_G (Q)$ and $\pow{Q} ( \pow{f}) \equiv 0 \pmod{x^\beta}$, then $Q(f) = 0$.
\end{lemma}
\begin{proof}
  Notice that since $f^tQ^{(t)} \in \Ya$ for all $t$, then $Q(f) \in \Ya$. Furthermore, since
  \[
    \delta(f^tQ^{(t)}) = \delta_{tG}(f^t) + \delta_{-tG}(Q^{(t)}) \le \delta_{-tG}(Q^{(t)}) \leq \delta_G(Q) \ ,
  \]
where the first inequality is due to $f \in \L(G)$, then $\delta(Q(f)) \leq \delta_G (Q)$. Combining this with the assumption that $\pow{Q}(\pow{f}) = \pow{Q(f)} \equiv 0 \pmod{x^\beta}$, we may conclude that $Q(f) \in \L(\delta_G(Q)\Pinf - \beta P_0)$, and if $\beta > \delta_G (Q)$, then this Riemann-Roch space is trivial.
\end{proof}

The second restriction on the precision $\beta$ is posed by the task of converting the truncated power series roots of $\pow{Q}$ back to $\L(G)$. Indeed, a basic root set $\{(\pow{f}_r,\alpha_r)\}_{r=1}^m$ describes each root $\pow{f}_r \in \field[x]$ of $\pow{Q}$ only to precision $\alpha_r$, and if this $\alpha_r$ is too small, then there could exist two distinct functions $h_1,h_2 \in \L(G)$ satisfying $\pow{h}_1 \equiv \pow{h}_2 \equiv \pow{f_r} \pmod{x^{\alpha_r}}$. In \cref{lem:high-enough-precision}, we will see how we can indirectly control $\alpha_r$ by increasing $\beta$; but first, let us show that conversion from truncated power series to $\L(G)$ is guaranteed to be unambiguous as long as $\alpha_r > \deg G$.

\begin{lemma}
  \label{lem:unique_intersect}
  If $\alpha > \deg G$, then for any $h \in \field[x]$ it holds that $|\L(G) \cap (h + x^\alpha\field[\![x]\!])| \leq 1$.
\end{lemma}
\begin{proof}
  If $h_1,h_2 \in \L(G) \cap (h + x^\alpha\field[\![x]\!])$, then $\pow{h_1} \equiv \pow{h_2} \equiv h \pmod{x^\alpha}$, which means that $h_1 - h_2 \in \L(G - \alpha P_0) = \{0\}$.
\end{proof}

Now we proceed by showing that the $\alpha_r$ from \cref{def:basic-root-set} can be made arbitrarily large by choosing the precision $\beta$ appropriately.

\begin{lemma}\label{lem:high-enough-precision}
  If $Q(z) = \sum_{t=0}^{\ell} z^t Q^{(t)} \neq 0$ with $Q^{(t)} \in \Ya(-tG)$, and if $f \in \Ya(G)$ satisfies $\pow{Q}(\pow{f} + x^{\alpha}z) \equiv 0 \pmod{x^\beta}$ for some $\alpha \in \ZZ$, then $\alpha \geq \frac{1}{\ell}(\beta - \delta_G(Q)) - \delta_G(f)$.
\end{lemma}
\begin{proof}
  We begin by defining
  \[
    T
    = Q(z + f)
    = \sum_{t=0}^{\ell} (z + f)^t Q^{(t)}
    = \sum_{t=0}^{\ell} \sum_{u=0}^{t} \binom{t}{u}z^u f^{t-u} Q^{(t)}
    = \sum_{u=0}^{\ell} z^u T_u \ ,
  \]
  where $T_u = \sum_{t=u}^{\ell} \binom{t}{u} f^{t-u} Q^{(t)}$. Since $f^{t-u} \in \Ya((t-u)G)$ and $Q^{(t)} \in \Ya(-tG)$, then $T_u \in \Ya(-uG)$. Furthermore, $x^{\alpha u}\pow{T}_u \equiv 0 \pmod{ x^{\beta}}$ for all $u$ because
  \[
    \pow{Q}(\pow{f} + x^{\alpha}z)
    = \pow{T}(x^{\alpha}z)
    = \sum_{u = 0}^{\ell}z^u x^{\alpha u} \pow{T}_u
    \equiv 0 \pmod{x^{\beta}} \ .
  \]
  Letting $r \in \{0,\dots,\ell\}$ be such that $\val[P_0](T_r) < \infty$ is maximal, observe that
  \begin{align*}
    \alpha \ell + \val[P_0](T_r) \geq \alpha r + \val[P_0](T_r) = \val[P_0](x^{\alpha r}T_r) \geq \beta \ ,
  \end{align*}
   which implies that $\alpha \geq \frac{1}{\ell}(\beta - \val[P_0](T_r))$.
   Finally, noting that
   \[
     0 \neq T_r \in \L(\delta_{-rG}(T_r)\Pinf - rG - \val[P_0](T_r)P_0) \ ,
   \]
   then the sought conclusion follows from
   \begin{align*}
    \val[P_0](T_r)
     &\le \delta_{-rG}(T_r) - r\deg G
     \le \delta_{-rG}(T_r) \\
     &= \delta_{-rG} \big(
     \sum_{t=r}^{\ell} \binom{t}{r} f^{t-r} Q^{(t)}
     \big) \\
     &\le \max_t \{ \delta_{-rG}(f^{t-r} Q^{(t)}) \} \\
     &\le \max_t \{ (t-r)\delta_{G}(f) + \delta_{-tG}(Q^{(t)}) \} \\
     &= \ell\delta_{G}(f) + \delta_G(Q) \ .
   \end{align*}
\end{proof}

Combining \cref{lem:high-enough-precision} and \cref{lem:unique_intersect}, we obtain the final restriction
\[
  \beta \ge 2\ell \deg G + s(n - \tau) \ ,
\]
which ensures that unambiguous conversion from the truncated power series roots of $\pow{Q}$ to $\L(G)$ is always possible. Indeed, this bound follows immediately from the fact that $\delta_G(f) \le \deg G$ for all $f \in \L(G)$ and the assumption that $\delta_G(Q) < s(n-\tau)$. Knowing that such conversion is possible, however, is not enough\,--\,we also need to know how to actually carry it out. In the following simple lemma, we show how to do this.

\begin{lemma}\label{lem:series-to-function}
  If $f \in \L(G)$ and $\sum_{i=0}^{\mu-1}f_i\pow{y}_i^{(G)} \equiv \hat{f} \pmod{x^\alpha}$ for some $f_i \in \field[x]$ with $\deg f_i \leq -\frac{1}{\mu}\delta_G(y_i^{(G)})$ and $\alpha > \deg G$, then $\sum_{i=0}^{\mu-1} f_i y_i^{(G)} = f$.
\end{lemma}
\begin{proof}
  Since $\delta (f_i) = \mu\deg f_i$, then $\delta_G(f_i y_i^{(G)}) \le \mu\deg f_i + \delta_G(y_i^{(G)}) \leq 0$. But then $\sum_{i=0}^{\mu-1}f_iy_i^{(G)} \in \L(G)$, and the conclusion follows from \cref{lem:unique_intersect}.
\end{proof}

Using the notation from \cref{def:hermit-pade} in the context of \cref{lem:series-to-function}, we see that $(f_0,\dots,f_{\mu-1},1) \in \H_{x^\alpha}(\mat{A})$, where $\mat{A} = [\pow{y}_0^{(G)},\cdots,\pow{y}_{\mu-1}^{(G)},-\pow{f}] \in \field[x]^{1 \times (\mu+1)}$.
Recovering $f \in \L(G)$ from $\pow{f} \rem x^{\alpha}$ thus translates to finding a polynomial vector $\vec{f} \in \H_{x^\alpha}(\mat{A})$ whose rightmost entry is $1$ and $\deg_{\vec{d}} \vec{f} = 0$, where
\[
  \vec{d} = \frac{1}{\mu}(\delta_G(y_0^{(G)}), \dots, \delta_G(y_{\mu-1}^{(G)}),0) \in (\tfrac{1}{\mu}\ZZ)^{\mu+1} \ .
\]
But this is easily accomplished by relying on \cref{thm:sim-hermit-pade} and \cref{cor:sim-hermit-pade}. We conclude this subsection by presenting our root-finding approach in its entirety in \cref{algo:find-roots}.

\begin{algorithm}[H]
  \caption{$\RootFind(D,G,Q,\pow{\vec{y}})$} \label{algo:find-roots}
  \begin{algorithmic}[1]
    \Input
    \begin{itemize}
    \item Divisors $D$ and $G$ for the code $\code$,
    \item a nonzero $Q=\sum_{t=0}^\ell z^tQ^{(t)} \in \M$ with $\delta_G(Q)<s(n-\tau)$, where $Q^{(t)}=\sum_{i=0}^{\mu-1}Q_i^{(t)}y_i^{(-tG)}$ for some $Q_i^{(t)} \in \field[x]$,
    \item $\pow{\vec{y}}=(\pow{y}_i^{(-tG)})_{i=0,\dots,\mu-1}^{t=0,\dots,\ell}$ with $\pow{y}_i^{(-tG)} \in \field[x]$ such that
      \newline $v_{P_0}(y_i^{(-tG)}-\pow{y}_i^{(-tG)}) \ge \beta := 2\ell\deg G+s(n-\tau)$.
    \end{itemize}
    \Output
    \begin{itemize}
    \item $L = \{f \in \L(G) \mid Q(f) = 0\}$ with $|L| \le \ell$.
    \end{itemize}
    \State $\pow{Q}^{(t)} \in \field[x] \assign \sum_{i=0}^{\mu-1}Q_i^{(t)}\pow{y}_i^{(-tG)}$ for $t=0,\dots,\ell$ 
    \State $\pow{Q} \in \field[x] \assign \sum_{t=0}^\ell z^t\pow{Q}^{(t)}$
    \State $\pow{L} \subset \field[x] \assign$ all polynomials from a basic root set of $\pow{Q}$ to precision $\beta$
    \State $L \assign \emptyset$
    \State $\vec{d} \in (\tfrac{1}{\mu}\ZZ)^{\mu+1} \assign \frac{1}{\mu}(\delta_G(y_0^{(G)}), \dots, \delta_G(y_{\mu-1}^{(G)}), 0)$
    \State $\alpha \in \ZZ_{>0} \assign \deg G+1$
    \For{$\pow{f} \in \pow{L}$} 
    \State $\mat F \in \field^{(\mu+1)\times(\mu+1)} \assign$ $\vec{d}$-Popov basis of $\H_{x^\alpha} ( [\pow{y}_0^{(G)},\dots ,\pow{y}_{\mu-1}^{(G)}, -\pow{f} ] )$
    \If{$\mat F$ contains a row $\vec{f} = (f_0,\dots,f_{\mu-1},1)$ with $\deg_{\vec{d}} \vec{f} = 0$}
    \State $L \assign L \cup \{\sum_{i=0}^{\mu -1} f_i y_i^{(G)}\}$
    \EndIf
    \EndFor
    \State \Return $L$
  \end{algorithmic}
\end{algorithm}

\begin{proposition}
  \cref{algo:find-roots} is correct and costs $\softO(\ell^2\mu^{\omega-1}(n+g))$ operations in $\field$.
\end{proposition}

\begin{proof}
  For correctness, our goal is to prove that $L = K$, where $L$ is the output of the algorithm and $K = \{f \in \L(G) \mid Q(f) = 0\}$. If $\{(\pow{f}_r,\alpha_r)\}_{r=1}^m \subset \field[x] \times \ZZ_{\ge 0}$ denotes the basic root set used in Line 3, i.e. $\pow{L} = \{\pow{f}_r\}_{r=1}^m$, then it is clear that $K \subseteq \bigcup_{r=1}^m K_r$, where $K_r = \L(G) \cap (\pow{f}_r + x^{\alpha_r}\field[\![x]\!])$ and $m \le \ell$. Since $\delta_G(Q)<s(n-\tau)$, $\delta_G(\pow{h}_r) \le \deg G$ and $\beta=2\ell \deg G+s(n-\tau)$, then \cref{lem:high-enough-precision} guarantees that $\alpha_r \ge \frac{1}{\ell}\big(\beta - \delta_G(Q)\big) - \delta_G(\pow{h}) \ge \deg G + 1$, hence $|K_r| \le 1$ by \cref{lem:unique_intersect}. Combining this with the fact that each non-empty $K_r$ necessarily contains an $\L(G)$-root of $Q$, as implied by \cref{lem:aprox-root-is-root} because $\beta > \delta_G(Q)$, we may conclude that $K = \bigcup_{r=1}^m K_r$. But due to \cref{lem:series-to-function}

  \begin{align*}
    \bigcup_{r=1}^m K_r = L
    = \Big\{ \sum_{i=0}^{\mu-1}f_i^{(r)}y_i^{(G)} \mid &\sum_{i=0}^{\mu-1}f_i^{(r)}\pow{y}_i^{(G)} \equiv \pow{f}_r \pmod{x^\alpha), \\ &\deg f_i^{(r)} \le - \frac{1}{\mu}\delta_G(y_i^{(G)}}, \ r=1,\dots,m \Big\} \ .
  \end{align*}

For the complexity, computing the $(\ell+1)\mu$ products $Q_i^{(t)}\pow{y}_i^{(-tG)}$ in Line 1 costs $\softO(\mu \ell\beta) \subseteq \softO(\ell^2\mu(n+g))$. The basic root set of $\pow{Q}$ in Line 3 can be computed with cost $\softO(\beta \deg_z(\hat{Q})) \subseteq \softO(\ell^2(n+g))$ due to \cite{neiger_fast_2017} (see \cref{thm:neiger-root-find}). Finally, the total cost of computing the $\vec{d}$-Popov bases in line 8 across all of the $\bigO(\ell)$ iterations in the surrounding for-loop is $\softO(\ell \mu^{\omega-1}(n+g))$ by \cref{cor:sim-hermit-pade}.
The claimed complexity of the algorithm follows.
 \end{proof}

\section{Decoding $\code$}
\label{sec:decoding}

We are now ready to state our Guruswami-Sudan list decoding algorithm for the code $\code$. We will assume that the decoding algorithm has access to the following data, which may be precomputed:
\begin{enumerate}
    \item divisor $E=E_1+\cdots+E_N$, where $E_1,\dots,E_N$ are distinct rational places different from $\Pinf$ not occurring in $\supp G$ and $N \ge \max\{\deg G+(\ell+3)(2g-1)+(s+1)n+2+\mu,(\ell+1)\deg G+4g+(s+1)n\}$,
    \item evaluations $\vec{g} = (g_{v,j}^{(u)})$, where $u=0,\dots,\ell$, $v = 1,2$ and $j = 1,\dots,N$, such that $g_{v,j}^{(u)} = g_v^{(u)}(E_j) \in \field$ where $\langle g_1^{(u)}, g_2^{(u)} \rangle_{\Ya} = \Ya(G_u)$, as in \cref{cor:Ya-basis-of-M}
    \item evaluations $\vec{x} = (x_j)_{j=1,\dots,N}$, where $x_j = x(E_j) \in \field$,
    \item evaluations $\vec{y} = (y_{i,j})^{i=0,\dots,\mu-1}_{j=1,\dots,N}$, where $y_{i,j} = \y[A]_i(E_j) \in \field$
    \item polynomials $\pow{\vec{y}}=(\pow{y}_i^{(-tG)})_{i,t} \in \field[x]^{\mu\times (\ell+1)}$, with $i=0,\dots,\mu-1$ and $t=-1,\dots,\ell$, polynomials in $\field[x]$ such that $v_{P_0}(y_i^{(-tG)}-\pow{y}_i^{(-tG)}) \ge 2\ell\deg G+s(n-\tau)$ for all $i$ and $t$,
\end{enumerate}

Then the decoding algorithm becomes the following:

\begin{algorithm}[H]
  \caption{$\Decode(\vec{r},s,\ell,D,G)$} \label{algo:gs-decode}
  \begin{algorithmic}[1]
    \Input
    \begin{itemize}
    \item Received word $\vec{r} \in \field^{n}$,
    \item divisors $D$ and $G$ for the code $\code$,
    \item decoding parameters $s,\ell \in \ZZ_{>0}$ with $s \le \ell$,
    \item corresponding list-decoding radius $\tau \in \ZZ_{>0}$,
    \end{itemize}
    \Output
    \begin{itemize}
    \item $L = \{f \in \L(G) \mid d(\vec{r},\vec{c}) \le \tau\}$ or $\algo{FAIL}$
    \end{itemize}
    \State $(B_{v,i}^{(u)})^{u=0,\dots,\ell}_{v=1,2, \ i=0,\dots,\mu-1} \assign \GenFqx(\vec{r},D,G,E,\vec{x},\vec{y},\vec{g})$
    \State ${\mat M}_{s,\ell}\in \field[x]^{2\mu(\ell+1)\times \mu(\ell+1)} \assign$ matrix based on the $B_{v,i}^{(u)}$ as in \cref{lem:Fx-mat}
    \State $\mat B_{s,\ell} \in \field[x]^{\mu(\ell+1)\times \mu(\ell+1)} \assign$ basis matrix in (unshifted) Popov form of ${\mat M}_{s,\ell}$
    \State $\vec{d} \in (\tfrac{1}{\mu}\ZZ)^{\mu(\ell + 1)} \assign (\vec{d}^{(0)}|\cdots|\vec{d}^{(\ell)})$ with $\vec{d}^{(t)} = \frac{1}{\mu}(\delta_{-tG}(y_{i}^{(-tG)}))_{i=0}^{\mu-1} \in (\tfrac{1}{\mu}\ZZ)^\mu$
    \State $\mat V_{s,\ell} \in \field[x]^{\mu(\ell+1)\times \mu(\ell+1)} \assign$ $\vec d$-Popov form of $\mat B_{s,\ell}$
    \State $\vec{Q} = \big( (Q_i^{(0)})_{i=0}^{\mu}|\dots|(Q_i^{(\ell)})_{i=0}^{\mu} \big) \in \field[x]^{\mu(\ell + 1)} \assign$ $\deg_{\vec{d}}$-minimal row of $\mat V_{s,\ell}$
    \If{$\deg_{\vec{d}} \vec{Q} \ge s(n-\tau)$}
    \State \Return $\algo{FAIL}$
    \EndIf
    \State $Q \in \bigoplus_{t=0}^\ell z^t \Ya(-tG) \assign \sum_{t=0}^\ell z^t\sum_{i=0}^{\mu-1}Q_i^{(t)}y_i^{(-tG)}$
    \State $L \assign \RootFind(D,G,Q,\pow{\vec{y}})$
    \For{$f \in L$}
    \State $\vec{c} \in \field^n \assign \Evaluate(f,D,G,\vec{x},\vec{y})$
    \If{$d(\vec{r}, \vec{c}) > \tau$}
    $L \assign L \setminus \{f\}$
    \EndIf
    \EndFor
    \State \Return $L$
  \end{algorithmic}
\end{algorithm}

Note that the decoding algorithm returns the functions from $\L(G)$ giving rise to all codewords within radius $\tau$ of the received word. Since in Line 12, the codeword corresponding to these function have been calculated, it is trivial to modify the algorithm to return these codewords instead. Combining all results from the previous section, we immediately obtain the following:

\begin{theorem}
The Guruswami-Sudan algorithm for the AG code $\code$ can be carried out in complexity $\softO(\ell^{\omega+1} \mu^{\omega-1}(n+g))$. Using the alternative generating set from \cref{rem:alternative-Fqx-basis-of-M}, we obtain the complexity $\softO(s\ell^\omega \mu^{\omega-1}(n+g))$.
\end{theorem}

We now give several examples comparing this result with previously known results.

\subsection{Examples}

\begin{example}
  AG codes obtained from the rational function field $\field(x)$ are known as \emph{generalized Reed-Solomon (GRS) codes}. In this case $g=0$ and $\mu=1$, which specializes the complexity of \cref{algo:gs-decode} to $\softO(s\ell^{\omega}n)$ operations in $\field$. The same complexity is achieved for families of function fields having fixed small genus, e.g. those arising from elliptic curves. The best known complexity for Guruswami-Sudan list-decoding of GRS codes is $\softO(s^2\ell^{\omega-1}n)$ \cite{chowdhury_faster_2015}.
\end{example}

\begin{example}
  By definition, any maximal function field $\ffield$ over $\field$ attains the Hasse-Weil bound\,--\,it has exactly $N_1 = q + 1 + 2g\sqrt{q}$ rational places, where $q$ is necessarily a square. If $\ffield$ is such a function field, then any place $P$ of $\ffield \fieldc$, where $\fieldc$ denotes the algebraic closure of $\field$, necessarily contains a positive element no larger than $\sqrt{q}$ in its Weierstrass semigroup \cite[Theorem 10.6]{HirKorchTor}, i.e. we are guaranteed that $\mu \le \sqrt{q}$ in the complexity of \cref{algo:gs-decode}. Furthermore, it is well known that all maximal function fields satisfy $g \le \sqrt{q}(\sqrt{q}-1)/2 \in \bigO(q)$. This implies that any code of length $n \in \Omega(q)$ over such a function field can be decoded using no more that $\softO(s\ell^\omega q^{(\omega-1)/2} n) \subseteq \softO(s\ell^\omega n^{(\omega+1)/2})$ operations in $\field$, which is sub-quadratic in the code length. Here, and in the rest of the examples, $u \in \Omega(v)$ if and only if $v \in \bigO(u)$ for any functions $u,v : \RR_{\ge 0} \to \RR_{\ge 0}$.
\end{example}

We obtain even better results for long codes over specific maximal function fields:

\begin{example}
  In the case of Hermitian function field $\ffield = \mathbb{F}_{q^2}(x_1,x_2)$, where $x_2^q + x_2 = x_1^{q+1}$, we have $N_1 = q^3 + 1$ rational places and genus $g=q(q-1)/2$. The usual choice of $\Pinf$ in one-point codes of $\ffield$ gives $\mu = q$. Consequently, we can decode any such code of length $n \in \Omega(q^3)$ using
  \[
    \softO(s\ell^\omega q^{\omega-1}q^3) = \softO(s\ell^\omega q^{\omega+2}) = \softO(s\ell^\omega n^{(\omega + 2)/3})
  \]
  operations in $\field$. For $n=q^3$, our approach specializes to the one from \cite{nielsen_sub-quadratic_2015}.
\end{example}

\begin{example}
  The Giulietti-Korchmaros function field $\mathbb{F}_{q^6}(x_1,x_2,x_3)$ from \cite{GiulKorch}, where $x_2^{q}+x_2=x_1^{q+1}$ and $x_3^{q^2-q+1}=x_1^{q^2}-x_1$, is also maximal\,--\,it has $\mu \le q^3$, $g=(q^5-2q^3+q^2)/2$ and $N_1=q^8-q^6+q^5+1$. In this case, we can decode any code of length $n \in \Omega(q^8)$ with cost $\softO(s\ell^{\omega}n^{(3\omega+5)/8})$.
\end{example}

\begin{example}
  The Suzuki function field $\ffield = \mathbb{F}_q(x_1,x_2)$, where $q = 2^{2e+1}$ is an odd power of two and $x_2^q+x_2=x_1^{2^e}(x_1^q+x_1)$, has genus $g=2^e(q-1)$ and $N_1=q^2+1$ rational places. Although it is not maximal in the sense of the Hasse-Weil bound, no other function field with the same genus and constant field can surpass its number of rational places \cite[Section 5.4]{Serre2020}. From \cite{BarMonZin}, it immediately follows that the Weierstrass semigroup of any place $P$ contains a positive element no greater than $q$, i.e. $\mu \le q$. This means that for any code over $\ffield$ of length $n \in \Omega(q^2)$, the complexity of \cref{algo:gs-decode} specializes to $\softO(s\ell^{\omega}n^{(\omega+1)/2})$.
\end{example}

\begin{example}
Let $F$ be a function field over $\field$ having a rational place $\Pinf$ whose Weierstrass semigroup can be generated by two positive integers, say $a$ and $b$, where $a<b$. Note that necessarily $\gcd(a,b)=1$, since otherwise the semigroup generated by $a$ and $b$ has infinitely many gaps. The genus of such a function field is $(a-1)(b-1)/2$, since this is the number of gaps of the semigroup generated by $a$ and $b$. Now let $x,y \in \Ya$ be such that $\delta(x)=a$ and $\delta(y)=b$. Then $F=\field(x,y)$ and $x^b+\alpha y^a+g(x,y)=0$, where $\alpha \in \field \setminus \{0\}$ and $g(X,Y) \in \field[X,Y]$ has $(a,b)$-weighted degree strictly less then $ab$. The curve defined by the equation $X^b+\alpha Y^a+g(X,Y)=0$ is sometimes called a $C_{ab}$-curve or a Miura-Kayima curve \cite{miura-kamiya-1993}; codes defined over such curves are of particular interest for practical applications, as they can be encoded efficiently \cite{beelen2020fast}. When it comes to decoding, the additional assumptions that $G=m\Pinf$ and that $D-n\Pinf$ is a principal divisor were used in \cite{beelen_efficient_2010} to decode the code $\code$ in complexity $\softO(\ell^5 a^3(n+g))$. 

Let us compare this to our results. Knowing that $F$ has a rational point $\Pinf$ whose Weierstrass semigroup contains two positive, relatively prime integers $a$ and $b$, implies that $g \le (a-1)(b-1)/2$ and $\mu \le a$. Using this weaker assumption and not needing the additional requirement that $G=m\Pinf$ and that $D-n\Pinf$ is a principal divisor, we can decode $\code$ in complexity $\softO(s\ell^\omega a^{\omega-1}(n+g))$. Hence, our results can both handle more general settings and decode faster.
\end{example}


\section*{Acknowledgments}

The authors would like to acknowledge the support from The Danish Council for Independent Research (DFF-FNU) for the project \emph{Correcting on a Curve}, Grant No.~8021-00030B.

\bibliographystyle{abbrv}
\bibliography{bibtex}

\end{document}

